\documentclass{article}
\usepackage{ijcai26}

\usepackage{times}
\usepackage{soul}
\usepackage{url}
\usepackage[hidelinks]{hyperref}
\usepackage[utf8]{inputenc}
\usepackage[switch]{lineno} 
\usepackage{amsmath}
\usepackage{amssymb}
\usepackage{mathtools}
\usepackage{amsthm}
\usepackage{amsfonts}

\usepackage{graphicx}
\usepackage{booktabs}
\usepackage{subcaption}
\usepackage{multirow}
\usepackage{array}
\usepackage{tabularx}
\usepackage{balance} 

\usepackage{xcolor}
\definecolor{maddpgcolor}{RGB}{76,114,176}
\definecolor{mappocolor}{RGB}{221,132,82}
\definecolor{comacolor}{RGB}{129,114,178}
\definecolor{qmixcolor}{RGB}{196,78,82} 
\definecolor{phiaccolor}{RGB}{85,168,104}

\usepackage{algorithm}
\usepackage{algpseudocode} 

\usepackage{microtype} 

\theoremstyle{plain}

\newtheorem{proposition}{Proposition}

\theoremstyle{definition}
\newtheorem{definition}{Definition}

\title{Phi-Actor-Critic: Steering General-Sum Games to \\ Pareto-Efficient Correlated Equilibria}

\author{
Wongyu Lee$^1$
\and
Francesco Lelli$^2$\and
Omran Ayoub$^3$\And
Massimo Tornatore$^1$\\
\affiliations
$^1$Politecnico di Milano\\
$^2$Tilburg University\\
$^3$University of Applied Sciences and Arts of Southern Switzerland\\
\emails
wongyu.lee@polimi.it,
f.lelli@tilburguniversity.edu,
omran.ayoub@supsi.ch,
massimo.tornatore@polimi.it
}

\begin{document}

\maketitle

\begin{abstract}
Real-world multi-agent systems, from traffic coordination to resource allocation, are often modeled as general-sum games where individual incentives conflict with collective welfare. In these settings, the central challenge is not merely finding an equilibrium, but selecting socially desirable outcomes among many suboptimal Nash equilibria. 
Standard deep multi-agent reinforcement learning (MARL) methods struggle with this problem, as value-decomposition approaches are constrained by monotonicity assumptions and policy-gradient methods often converge to stable but socially inefficient equilibria.
To address this limitation, we propose $\Phi$-Actor-Critic ($\Phi$-AC), a framework that leverages swap regret minimization to steer learning toward high-welfare correlated equilibria (CE). To make counterfactual regret estimation tractable in deep MARL, $\Phi$-AC employs a centralized attention critic that predicts vector-valued regrets in a single forward pass, avoiding computationally expensive counterfactual simulations. We further introduce a Lagrangian-based equilibrium selection mechanism that optimizes social welfare while enforcing stability through regret constraints.
Experiments on matrix games, Multi-Agent Particle Environments (MPE), and the Melting Pot Harvest scenario demonstrate that $\Phi$-AC learns efficient and stable coordination strategies across diverse mixed-motive settings while maintaining high collective return and competitive fairness.
\end{abstract}

\section{Introduction}
Achieving socially efficient coordination, where agents collectively obtain higher welfare, remains challenging in multi-agent systems with partially conflicting incentives. 
Even when higher-welfare collective strategies exist, independent learning dynamics in mixed-motive environments often converge to individually rational but socially inefficient behaviors, such as risk-averse coordination or persistent asymmetric roles \cite{shoham2008multiagent}.
Existing deep multi-agent reinforcement learning (MARL) methods often struggle to avoid such outcomes. 
Value-decomposition approaches such as QMIX \cite{rashid2020monotonic} rely on restrictive monotonicity assumptions, while policy-gradient methods (PGMs) such as MAPPO \cite{yu2022surprising} lack mechanisms that explicitly select socially desirable equilibria.
As a result, these methods may converge to stable but low-welfare Nash equilibria (NE).

To address this limitation, we shift the focus from standard reward maximization to regret-based equilibrium selection rooted in algorithmic game theory (AGT) \cite{hart2000simple,blum2007external,cesa2006prediction}. 
Rather than targeting only independent NE, regret-based learning enables correlated solution concepts that support richer coordination in general-sum games. 
In particular, while external regret minimization is associated with broad coarse correlated equilibria (CCE), swap-regret-based learning refines the target to the tighter class of correlated equilibria (CE) by evaluating whether an agent would benefit from conditionally replacing one action with another \cite{hart2000simple,blum2007external}.

Based on this insight, we propose $\Phi$-Actor-Critic ($\Phi$-AC), a regret-aware framework designed to operationalize swap regret for equilibrium selection in deep MARL. 
We address two practical challenges in applying swap regret to deep MARL: scalable regret estimation and equilibrium selection.
First, since calculating swap regret typically requires computationally expensive counterfactual evaluation, we introduce a centralized attention critic that efficiently predicts vector-valued regrets in a single forward pass ($O(1)$). 
Second, since low regret only encourages behavior consistent with approximate CE conditions, it does not specify which equilibrium within the CE set should be selected. 
We therefore introduce the Regret-Balancing Social Welfare Objective (RB-SWO), which biases learning toward Pareto-efficient CE by optimizing social welfare under regret constraints.
By formulating the learning problem as a constrained Markov game optimized via a Lagrangian objective, $\Phi$-AC encourages socially efficient behavior while maintaining low swap regret.

Our contributions are threefold:
\begin{enumerate}
    \item \textbf{Efficient Regret Estimation.} We introduce a regret-conditioned critic that predicts swap regret in a single forward pass, avoiding expensive counterfactual simulations in high-dimensional environments.
    \item \textbf{Principled Equilibrium Selection.} We introduce a Lagrangian-based selection mechanism that steers learning toward efficient and fair equilibria under regret constraints.
    \item \textbf{Scalable Coordination in Social Dilemmas.} We demonstrate that $\Phi$-AC learns high-welfare coordination strategies in matrix games and scales to complex Sequential Social Dilemmas (SSDs), outperforming strong MARL baselines in both sustainability and fairness.
\end{enumerate}

\section{Related Work}
\label{sec:related_work}
\textbf{Scalable Deep MARL.}
Centralized-training decentralized-execution (CTDE) has become a standard paradigm for scalable MARL, allowing methods such as MADDPG \cite{lowe2017multi} and COMA \cite{foerster2018counterfactual} to use centralized critics during training while preserving decentralized execution. 
Within this paradigm, COMA improves credit assignment through counterfactual baselines, while value-decomposition methods such as QMIX \cite{rashid2020monotonic} further improve scalability by factorizing the joint action-value function through a monotonic mixing assumption.
Recent PGM such as MAPPO \cite{yu2022surprising} also demonstrate strong empirical performance across cooperative MARL benchmarks.
Despite these advances, these methods remain primarily optimized for expected return. 
QMIX is further constrained by its monotonic mixing assumption, which can be restrictive in mixed-motive games, while PGMs do not explicitly impose regret-based equilibrium selection. 
Thus, they may achieve stable learning or high returns without controlling whether the selected outcome is socially desirable.

\noindent\textbf{Regret-Based Learning.}
Regret-minimization methods provide a principled route to equilibrium computation. 
Counterfactual Regret Minimization (CFR) and its neural variants, including Deep CFR and DREAM, minimize counterfactual regret and provide convergence guarantees primarily in two-player zero-sum imperfect-information games \cite{zinkevich2007regret,brown2019deep,steinberger2020dream}. 
Similarly, NeuRD connects policy-gradient learning to Hedge \cite{freund1997decision} and replicator dynamics in settings where the equilibrium objective is commonly formulated as approximate Nash convergence \cite{hennes2020neural}. 
In contrast, $\Phi$-AC targets general-sum MARL, where swap-regret-inspired learning naturally relates to correlated-equilibrium conditions and where the main challenge is selecting efficient CE rather than computing a zero-sum Nash solution.

\noindent\textbf{Positioning of $\Phi$-AC.}
$\Phi$-AC bridges these lines of work by combining scalable centralized-critic MARL with swap-regret-based equilibrium selection. 
Rather than only improving return maximization or Nash convergence, $\Phi$-AC uses learned swap-regret signals to steer policies toward Pareto-efficient CE while remaining applicable to high-dimensional mixed-motive environments.


\section{Preliminaries}
\label{sec:preliminaries}
We formulate the environment as a general-sum Markov game and study equilibrium selection through the lens of CE.

\subsection{General-Sum Markov Games}
\label{sec:mg_definition}
We model the multi-agent environment as a General-Sum Markov Game (MG), formally defined by the tuple
$\mathcal{G} =\langle\mathcal{N},\mathcal{S},\{\mathcal{A}_i\}_{i\in\mathcal{N}},
\mathcal{P},\{r_i\}_{i\in\mathcal{N}},\gamma,\Omega,\mathcal{O}\rangle.$
Here, $\mathcal{N}=\{1,\dots,N\}$ is the set of agents, $\mathcal{S}$ is the global state space, $\mathcal{P}:\mathcal{S}\times\mathcal{A}\rightarrow\Delta(\mathcal{S})$ denotes the state transition function, $r_i:\mathcal{S}\times\mathcal{A}\rightarrow\mathbb{R}$ is the reward function for agent $i$, and $\gamma\in[0,1)$ is the discount factor. 
We assume a partially observable setting where each agent $i$ receives a local observation $o_i \in \Omega_i$ through the observation function $\mathcal{O}:\mathcal{S}\times\mathcal{N}\rightarrow\Omega_i$.
The joint action is denoted by $\mathbf{a}=(a_i,\mathbf{a}_{-i})\in\mathcal{A}\equiv\prod_i\mathcal{A}_i$.

\subsection{Correlated Equilibrium and OSDP}
\label{sec:ce_osdp}

While NE assumes independent policies, CE \cite{aumann1974subjectivity} allows coordination via a correlation device, often enabling higher-welfare outcomes. 
In Markov games, however, equilibrium conditions involve deviations over sequential decision processes rather than a single normal-form interaction. 
The One-Shot Deviation Principle (OSDP) states that, under standard dynamic-game conditions, it is sufficient to check whether any agent can profit by deviating at a single decision point while following the original policy thereafter \cite{hendon1996one}. 
We use this principle as a local approximation of sequential deviation incentives.
In our setting, this amounts to measuring whether an agent can improve its return by unilaterally changing its current action, using the stationary Q-function $Q_i^{\boldsymbol{\pi}}(s,\mathbf{a})$ as a local estimate of deviation incentives.
We utilize the framework of $\Phi$-Equilibria \cite{greenwald2003general}. 
Let $\Phi_i$ be a set of deviation functions $\phi: \mathcal{A}_i \to \mathcal{A}_i$ for agent $i$, where each deviation maps an action $a_i$ to an alternative action $\phi(a_i)$.
\begin{definition}[$\Phi$-Equilibrium {\cite{greenwald2003general}}]
A joint policy $\boldsymbol{\pi}$ constitutes a $\Phi$-Equilibrium at state $s$ if no agent $i$ can improve their utility by applying any deviation $\phi \in \Phi_i$:
\begin{multline}
    \sum_{\mathbf{a} \in \mathcal{A}} \boldsymbol{\pi}(\mathbf{a}|s) Q_i^{\boldsymbol{\pi}}(s, \mathbf{a}) \geq 
    \sum_{\mathbf{a} \in \mathcal{A}} (\phi \circ_i \boldsymbol{\pi})(\mathbf{a}|s) Q_i^{\boldsymbol{\pi}}(s, \mathbf{a}), \\
    \forall i \in \mathcal{N}, \forall \phi \in \Phi_i.
    \label{eq:phi_eq_definition}
\end{multline}
\end{definition}
Here, $(\phi \circ_i \boldsymbol{\pi})$ denotes the joint distribution where agent $i$ follows the deviation $\phi$ while others follow $\boldsymbol{\pi}$.
When $\Phi_i$ includes all possible swap deviations, Eq.~\eqref{eq:phi_eq_definition} characterizes CE.

\subsection{Convergence via No-Swap-Regret Dynamics}
\label{sec:no_phi_regret}
A fundamental result in AGT links regret minimization to equilibrium convergence \cite{blum2007external,cesa2006prediction}. 
We define the instantaneous swap regret vector following standard swap-regret formulations \cite{hart2000simple,blum2007external}.
$\boldsymbol{\mathcal{R}}_{i}^{\Phi}(s, \mathbf{a}) \in \mathbb{R}^{|\mathcal{A}_i|}$ where each element corresponds to the counterfactual gain of swapping the chosen action $a_i$ to a specific alternative $a'_i \in \mathcal{A}_i$:
\begin{equation}
    \boldsymbol{\mathcal{R}}_{i}^{\Phi}(s, \mathbf{a})[a'] = \left[ Q_i^{\boldsymbol{\pi}}(s, (a', \mathbf{a}_{-i})) - Q_i^{\boldsymbol{\pi}}(s, \mathbf{a}) \right]^+.
    \label{eq:inst_regret_def}
\end{equation}
where $[\cdot]^+$ denotes the ReLU operator. This vector formulation allows the critic to estimate regret for all possible deviations simultaneously, which is crucial for scalable optimization.

\begin{proposition}[No-$\Phi$-Regret to Approximate $\Phi$-Equilibrium 
\cite{hart2000simple,greenwald2003general,blum2007external}]
\label{prop:no_phi_regret}
Under standard finite repeated-game dynamics, let $x_t \in \Delta(\mathcal{A})$ be the joint action distribution induced at round $t$, and let $u_i(x_t)$ denote agent $i$'s expected utility under $x_t$. Define
$R_T^{\Phi,i}
=
\max_{\phi \in \Phi_i}
\sum_{t=1}^T
\left[
u_i(\phi \circ_i x_t) - u_i(x_t)
\right].$
If every agent satisfies $R_T^{\Phi,i}/T \leq \epsilon$, then the empirical distribution
$\bar{x}_T = \frac{1}{T}\sum_{t=1}^T x_t$
is an $\epsilon$-$\Phi$-equilibrium. 
When $\Phi_i$ contains all swap deviations, $\bar{x}_T$ is an $\epsilon$-CE.
\end{proposition}

\paragraph{The Selection Gap.} 
While Proposition~\ref{prop:no_phi_regret} establishes that minimizing swap regret leads to approximate CE behavior, it does not determine which equilibrium is selected. In general-sum games, some equilibria may still be socially inefficient or unfair, necessitating an explicit selection mechanism. In Section~\ref{sec:methods}, we address this challenge via a constrained Markov game formulation using a Lagrangian objective that favors socially efficient equilibria.


\section{Methodology: $\Phi$-Actor-Critic}
\label{sec:methods}
We propose $\Phi$-AC, a deep MARL framework designed to encourage efficient coordination through swap-regret-based learning.
The framework consists of three components: scalable regret estimation (Sec.~\ref{sec:critic_architecture}), constrained equilibrium selection (Sec.~\ref{sec:lagrangian_selection}), and exploration dynamics (Sec.~\ref{sec:entropy_dynamics}). The overall centralized training optimization orchestrating these components is summarized in Algorithm~\ref{alg:phi_ac_optimization}.

\subsection{Scalability: The Regret-Conditioned Critic}
\label{sec:critic_architecture}

\paragraph{From Simulation to Prediction.} 
The theoretical convergence to CE relies on calculating \textit{Instantaneous Swap Regret} (Eq.~\ref{eq:inst_regret_def}). In tabular settings with a simulator, one can simply reset the environment to query ``what if'' for every unchosen action ($O(N|\mathcal{A}|)$ complexity). However, in real-time deployment where resetting is impossible, explicitly computing these counterfactual rewards becomes prohibitively expensive.

$\Phi$-AC solves this bottleneck by training a centralized attention mechanism inspired by MAAC \cite{iqbal2019actor} to explicitly \textit{predict} these counterfactual values. This provides dense regret estimates through a single forward pass, improving scalability in high-dimensional environments.
As depicted in Figure~\ref{fig:Phi_AC_structure}, we design a dual-head critic to jointly model value estimation and swap-regret prediction.
To account for non-stationary learning dynamics induced by concurrently adapting agents, we condition the critic on running cumulative regret $\boldsymbol{\mathcal{R}}^{\mathrm{cum}}$ via a Feature-wise Linear Modulation (FiLM) layer, which dynamically scales and shifts critic features \cite{perez2018film}.

The critic consists of two heads:
\begin{itemize}
    \item \textbf{Q-Head:} Estimates action-values $Q^\psi_i(s, \mathbf{a}, \boldsymbol{\mathcal{R}}^{\text{cum}})$.
    \item \textbf{Regret-Head:} Predicts the regret vector $\hat{\boldsymbol{\mathcal{R}}}^\psi_i \in \mathbb{R}^{|\mathcal{A}_i|}$, where each element approximates the counterfactual gain of switching from $a_i$ to $a'_i$.
\end{itemize}
\noindent To stabilize regret estimation, we compute the regression target $\boldsymbol{\mathcal{R}}^{\text{target}}_i$ using a Target Network $\psi'$ with parameters lagging $\psi$:
\begin{equation}
    \boldsymbol{\mathcal{R}}^{\text{target}}_i(a') = \left[ Q^{\psi'}_i(s, (a', \mathbf{a}_{-i})) - Q^{\psi'}_i(s, \mathbf{a}) \right]^+.
\end{equation}
The critic minimizes the joint loss, including 

$\mathcal{L}_R = \frac{1}{N} \sum_i \|  \boldsymbol{\mathcal{R}}^{\text{target}}_i-\hat{\boldsymbol{\mathcal{R}}}^\psi_i \|_2^2$.

\begin{algorithm}[tb]
\small
\caption{Centralized Training Optimization in $\Phi$-AC}
\label{alg:phi_ac_optimization}
\begin{algorithmic}[1]
\State \textbf{Input:} Minibatch $\mathcal{B} = (\mathbf{o}, \mathbf{a}, \mathbf{r}, \mathbf{o}')$ from $\mathcal{D}$, Current Cumulative Regret $\boldsymbol{\mathcal{R}}^{\text{cum}}$.
\State \textbf{Hyperparameters:} $\delta_{\text{regret}}$ (Tolerance), $\delta$ (EMA Decay rate), $\eta_\alpha$ (Dual LR).

\Statex \textit{// 1. Regret-Conditioned Critic Update} (Sec.~\ref{sec:critic_architecture})

\State Calculate TD targets $y_i$ and regret targets $\boldsymbol{\mathcal{R}}^{\text{target}}_i$ using target networks $\psi', \theta'$.

\State Get outputs from dual-head Critic conditioned on $\boldsymbol{\mathcal{R}}^{\text{cum}}$ via FiLM:
\State \hspace{\algorithmicindent} $\{Q^{\psi}_i, \hat{\boldsymbol{\mathcal{R}}}^{\psi}_i\} \leftarrow \text{Critic}(\mathbf{o}, \mathbf{a}, \boldsymbol{\mathcal{R}}^{\text{cum}})$

\State Update Critic $\psi$ by minimizing the joint loss:
\State \hspace{\algorithmicindent} $\mathcal{L}_{\text{Critic}} = \frac{1}{N} \sum_i \Big( \mathcal{L}_Q(Q^{\psi}_i, y_i) + \mathcal{L}_R(\hat{\boldsymbol{\mathcal{R}}}^{\psi}_i, \boldsymbol{\mathcal{R}}^{\text{target}}_i) \Big)$

\Statex \textit{// 2. Actor \& Dual Update via RB-SWO} (Sec.~\ref{sec:lagrangian_selection})

Sample differentiable actions $\tilde{\mathbf{a}} \sim \boldsymbol{\pi}_{\theta}(\mathbf{o})$ via Gumbel-Softmax using regret-biased logits (Eq.~\eqref{eq:Regret_Matching})

\State Get Critic predictions $\{Q^{\psi}_i(\tilde{\mathbf{a}}), \hat{\boldsymbol{\mathcal{R}}}^{\psi}_i(\tilde{\mathbf{a}})\}$ evaluated at the sampled actions.

\State Calculate Regret Magnitudes $\mathcal{M}_i \leftarrow \| [\hat{\boldsymbol{\mathcal{R}}}^{\psi}_i(\tilde{\mathbf{a}})]^+ \|_2$.
\State \textbf{[Primal]} Update Actor $\theta$ by minimizing:
\State \hspace{\algorithmicindent} $\mathcal{L}_{\text{total}} \leftarrow \mathcal{L}_{\text{RB-SWO}}(\theta,\boldsymbol{\alpha})$ (Eq.~\eqref{eq:lagrangian_loss})

\State \textbf{[Dual]} Update multipliers $\alpha_i^{\text{fair}}, \alpha_i^{\text{ent}}$ via gradient ascent:
\State \hspace{\algorithmicindent} $\alpha_i^{\text{fair}} \leftarrow \left[ \alpha_i^{\text{fair}} + \eta_\alpha (\mathcal{M}_i - \delta_{\text{regret}}) \right]^+$
\State \hspace{\algorithmicindent} $\alpha_i^{\text{ent}} \leftarrow \left[ \alpha_i^{\text{ent}} + \eta_\alpha (\mathcal{H}^{\text{target}}_i - \mathcal{H}(\pi_{\theta_i})) \right]^+$

\Statex \textit{// 3. Cumulative Regret \& Target Networks Update} (Sec.~\ref{sec:entropy_dynamics})

\State Update Cumulative Regret using Exponential Moving Average (EMA):
\State \hspace{\algorithmicindent} $\boldsymbol{\mathcal{R}}^{\text{cum}}_i \leftarrow \delta \boldsymbol{\mathcal{R}}^{\text{cum}}_i + (1 - \delta) \hat{\boldsymbol{\mathcal{R}}}^{\psi}_i(\mathbf{a})$
\State Soft update target networks $\psi'$ and $\theta'$.
\end{algorithmic}
\end{algorithm}

\subsection{Selection: Lagrangian Dual Optimization}
\label{sec:lagrangian_selection}
Minimizing swap regret encourages stability, but it does not determine which equilibrium is selected when multiple low-regret equilibria exist. 
To address this selection gap, we introduce the Regret-Balancing Social Welfare Objective (RB-SWO), which maximizes collective welfare while constraining each agent's regret magnitude. 
This formulation favors high-welfare coordination without allowing the policy to move far from the low-regret region.

\paragraph{Formulation as Constrained Optimization.}
We formalize RB-SWO as a constrained optimization problem over expected discounted returns.
We explicitly define the \textit{Regret Magnitude} $\mathcal{M}_i(\boldsymbol{\pi})$ as the expected $L_2$-norm of the positive swap regret vector:
\begin{equation}
    \mathcal{M}_i(\boldsymbol{\pi}) = \mathbb{E} \left[ \left\| \left[ \hat{\boldsymbol{\mathcal{R}}}^\psi_i(s, \mathbf{a}) \right]^+ \right\|_2 \right]
\end{equation}
We also define the policy entropy as 
$\mathcal{H}(\pi_i)=\mathbb{E}_{a_i\sim\pi_i}[-\log \pi_i(a_i|o_i)]$, 
with $\mathcal{H}^{\mathrm{target}}_i(t)$ denoting a time-varying minimum entropy target.

\noindent The optimization problem is:
\begin{align}
    \max_{\boldsymbol{\pi}} \quad 
    & J(\boldsymbol{\pi}) 
    = \mathbb{E}_{s \sim d^{\boldsymbol{\pi}}, \mathbf{a} \sim \boldsymbol{\pi}}
    \left[ \sum_{i \in \mathcal{N}} Q_i^{\boldsymbol{\pi}}(s,\mathbf{a}) \right] \\
    \text{s.t.} \quad 
    & \mathcal{M}_i(\boldsymbol{\pi}) \leq \delta_{\text{regret}}, \quad \forall i \in \mathcal{N} \\
    & \mathcal{H}(\pi_i) \geq \mathcal{H}^{\text{target}}_i(t), \quad \forall i \in \mathcal{N}.
\end{align}

\begin{figure}[t]
    \centering
    \includegraphics[width=0.495\textwidth]{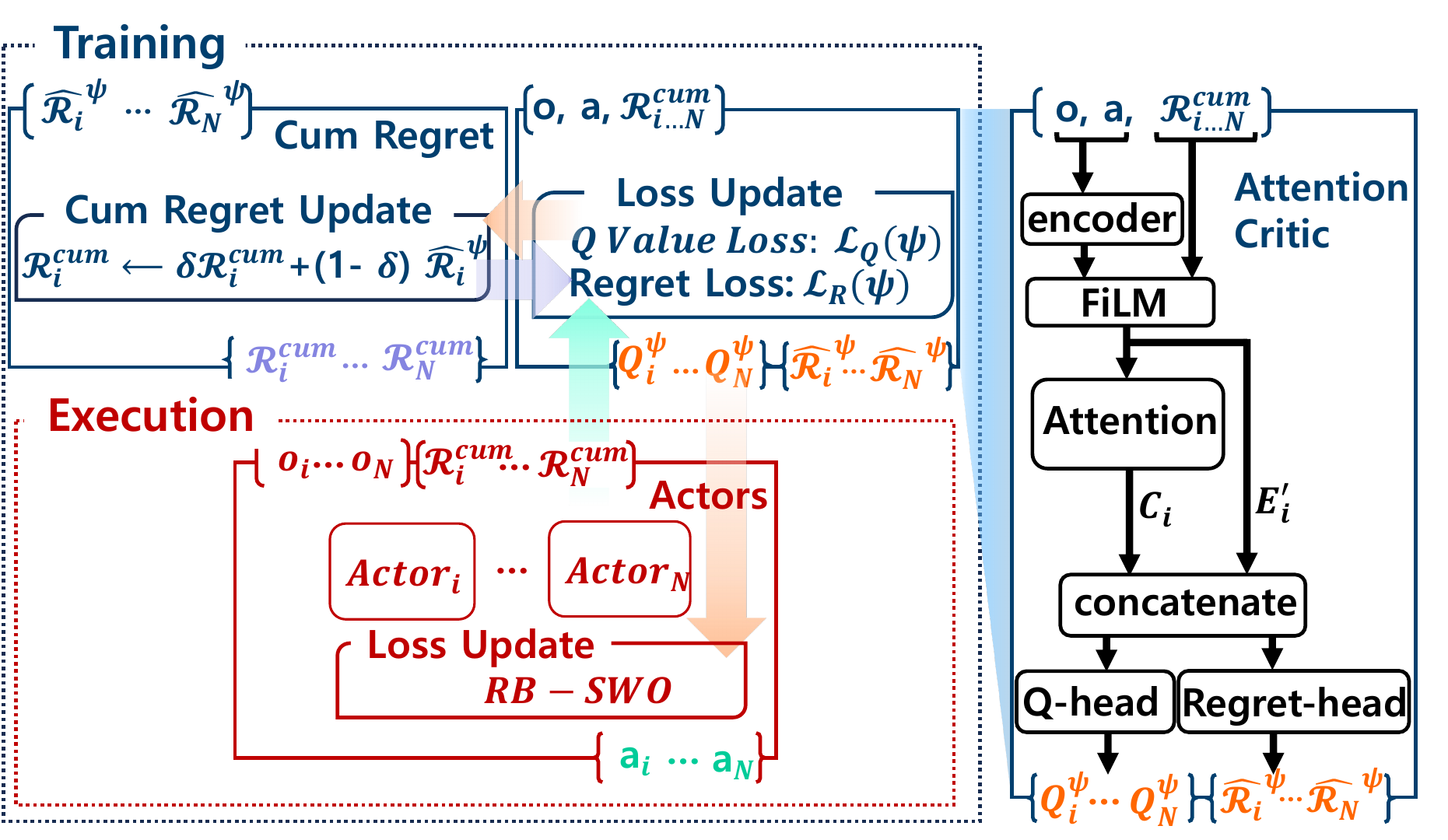}
    \caption{
        The architecture of $\Phi$-AC. 
        (Left) Decentralized actors select actions using local observations and cumulative regret bias. 
        (Right) The centralized attention critic processes global states and cumulative regret to output both Q-values and explicit regret predictions. 
        (Center) The Regret Coordinator modulates the critic via FiLM based on the cumulative regret feedback loop.}
    \label{fig:Phi_AC_structure}
\end{figure}
where $\delta_{\text{regret}}$ controls the allowed regret magnitude, and $\mathcal{H}^{\mathrm{target}}_i(t)$ is annealed over training.
We employ the Gumbel-Softmax reparameterization trick \cite{jang2016categorical} to enable differentiable optimization of the regret constraint.

\paragraph{Lagrangian Dual Optimization.}
We solve this via the primal-dual method. Introducing the dual variables $\boldsymbol{\alpha}=\{\alpha_i^{\text{fair}},\alpha_i^{\text{ent}}\}_{i=1}^N,$
where $\alpha_i^{\text{fair}} \ge 0$ controls regret-based regularization and $\alpha_i^{\text{ent}} \ge 0$ encourages exploration, we minimize the following objective for actor parameters $\theta$:
\begin{equation}
    \begin{split}
        \mathcal{L}_{\text{total}}(\theta, \boldsymbol{\alpha}) = \mathbb{E} \Bigg[ 
        & \underbrace{- \sum_{i} Q^{\boldsymbol{\pi}}_i}_{\text{Max Welfare}} 
        + \sum_{i} \alpha_i^{\text{fair}} (\mathcal{M}_i - \delta_{\text{regret}}) \\
        & + \sum_{i} \alpha_i^{\text{ent}} 
        ( \mathcal{H}^{\text{target}}_i - \mathcal{H}(\pi_i) ) 
        \Bigg].
    \end{split}
    \label{eq:lagrangian_loss}
\end{equation}

The optimization proceeds as an alternating game: 
\begin{itemize}
    \item \textbf{Primal Update (Actor):} Update $\theta$ to minimize $\mathcal{L}_{\text{total}}$.
    \item \textbf{Dual Update (Multipliers):} Update multipliers via gradient ascent with step size $\eta_{\alpha} > 0$: 
    \begin{equation}
        \alpha_i^{\text{fair}} \leftarrow \left[ \alpha_i^{\text{fair}} + \eta_{\alpha} (\mathcal{M}_i - \delta_{\text{regret}}) \right]^+.
    \end{equation}
\end{itemize}
When regret increases, the corresponding dual variable also increases, placing greater emphasis on stability during optimization.

\subsection{Dynamics: Entropy Annealing \& Warm-up}
\label{sec:entropy_dynamics}

A practical challenge in regret-based learning is that agents cannot regret actions they have never explored.
To resolve this, we implement an explicit schedule:
\paragraph{Discovery (Warm-up).} For the first $T_{\text{warm}}$ steps (a predefined hyperparameter), we suppress the fairness penalty ($\alpha^{\text{fair}} \to 0$) and enforce a high entropy target. This forces agents to explore the state space purely via Maximum Entropy RL (MaxEnt RL) \cite{haarnoja2018soft}, accumulating the reward experiences necessary to define well-defined regret.
\paragraph{Convergence (Annealing).} Post warm-up, we activate the regret constraints and linearly decay the entropy, encouraging diverse exploration before convergence to stable coordinated behaviors.


\section{Theoretical Analysis}
\label{sec:theoretical_analysis}

In this section, we provide a theoretical interpretation of $\Phi$-AC by connecting its policy update to no-swap-regret learning and by explaining how the proposed Lagrangian objective biases equilibrium selection. Rather than relying on exact tabular regret computation, $\Phi$-AC uses a learned critic to approximate regret signals in a differentiable deep MARL setting.

\subsection{Approximating Regret Matching via Deep Learning}
\label{sec:theory_approximation}
We next interpret the actor update in $\Phi$-AC through the lens of regret matching.

\paragraph{Actor as Smooth Regret Matching.}
Classically, convergence to CE can be achieved via Regret Matching (RM) \cite{hart2000simple}, where agents sample actions proportionally to their positive regret values, i.e., $\pi(a_i) \propto [\boldsymbol{\mathcal{R}}^{\Phi}_i(a_i)]^+$. However, the non-differentiability of the standard RM operator makes it unsuitable for gradient-based optimization in deep networks.

To bridge this gap, $\Phi$-AC operationalizes RM via a Regret-Biased Softmax Policy. By injecting the predicted cumulative regret $\hat{\boldsymbol{\mathcal{R}}}^{\text{cum}}$ into the logits, the policy becomes:
\begin{equation}
    \pi_{\theta}(a_i|{o_i}) = \frac{\exp\left( (L_i({o_i},a_i) + \beta [\hat{\boldsymbol{\mathcal{R}}}^{\text{cum}}_i(a_i)]^+ ) / \tau \right)}{\sum_{a'} \exp\left( (L_i({o_i},a') + \beta [\hat{\boldsymbol{\mathcal{R}}}^{\text{cum}}_i(a')]^+ ) / \tau \right)}.
    \label{eq:Regret_Matching}
\end{equation}
Here, $L_i(o_i, a_i)$ denotes the raw logits output by the actor network prior to softmax normalization, $\tau$ is the temperature, and $\beta$ controls the strength of the regret bias. 
Lower temperatures make the policy more sensitive to positive regret, while larger $\beta$ places greater weight on actions with accumulated regret. 
Together, these terms yield a smooth and differentiable approximation to regret-driven action selection.

This formulation can be interpreted as a smooth approximation of RM. While standard RM selects actions in proportion to positive regret, our softmax-based formulation is conceptually related to exponential weights and Hedge algorithms \cite{freund1997decision,cesa2006prediction}, which are also known to induce no-regret behavior under standard online learning assumptions.

\paragraph{Approximate Regret under Function Approximation.}
In deep MARL, regret estimates are inherently approximate due to function approximation. 
Following standard analyses of perturbed no-regret dynamics \cite{cesa2006prediction}, we consider the setting where the centralized critic provides uniformly bounded regret estimation error:
$
\max_{s,\mathbf{a}}
\|
\boldsymbol{\mathcal{R}}^{\mathrm{true}}
-
\hat{\boldsymbol{\mathcal{R}}}^{\psi}
\|_\infty
\leq \epsilon_c.
$
Under policy updates that are sufficiently smooth with respect to the regret estimates, the induced perturbation scales with $\epsilon_c$, implying that the resulting average swap regret remains bounded up to an $\mathcal{O}(\epsilon_c)$ approximation term.
Consequently, the learned dynamics may approach an approximate CE despite imperfect regret estimation.

\subsection{Geometry of Equilibrium Selection}
\label{sec:theory_selection}
While regret minimization encourages behavior consistent with approximate CE conditions, it does not by itself distinguish high-welfare equilibria from inefficient low-regret outcomes. 
RB-SWO addresses this selection gap by adding welfare maximization under individual regret constraints, yielding a Lagrangian objective that favors high-welfare coordination.

\paragraph{Dual Forces in Optimization.}
We formulate equilibrium selection as maximizing social welfare ($W = \sum Q_i$) subject to stability constraints ($\mathcal{M}_i \le \delta_{\text{regret}}$). The Lagrangian objective $\mathcal{L}(\theta, \alpha) = W - \sum \alpha_i (\mathcal{M}_i - \delta_{\text{regret}})$ induces two distinct gradient forces:

\noindent\textbf{Efficiency Force ($\nabla_{\theta} W$):} This term drives the joint policy along the ``flat" directions of the CE polytope (where regret is zero) toward higher-welfare regions of the CE set.

\noindent\textbf{Restoring Force ($-\alpha_i \nabla_{\theta} \mathcal{M}_i$):} If an agent deviates from equilibrium (increasing $\mathcal{M}_i$), the dual variable $\alpha_i$ increases (Dual optimization). This increases the optimization pressure toward lower-regret regions.

\paragraph{Fairness via Individual Constraints.}
Individual regret constraints also play an important role in equilibrium selection.
Geometrically, asymmetric equilibria (where one agent yields and another benefits) imply that while the sum of regrets might be low, the individual regret for the yielding agent is high. 
By enforcing $\mathcal{M}_i \le \delta$ for each agent rather than only controlling the average regret, RB-SWO discourages solutions in which one agent absorbs most of the deviation incentive. This encourages more balanced outcomes and reduces winner-takes-all behavior.

\section{Experimental Evaluation}
\label{sec:experiments}
While the previous sections introduced the regret-aware architecture and optimization framework of $\Phi$-AC, this section empirically evaluates its ability to learn socially efficient coordination. 
We evaluate $\Phi$-AC across a hierarchy of complexity: diagnostic matrix games, scalable MPE domains, and Melting Pot Harvest. 
These settings respectively test equilibrium selection in transparent games, scalability in higher-dimensional interactions, and sustainable coordination under shared resource depletion.

We compare $\Phi$-AC against representative MARL baselines that capture different algorithmic design choices. 
MADDPG represents deterministic policy-gradient learning, MAPPO and COMA represent stochastic policy-gradient methods with centralized training, and QMIX represents value decomposition under monotonic factorization.

\paragraph{Experimental Setup.}
Training configurations were tailored to the complexity of each domain: Iterated Matrix Games (IMG) were trained for 1k episodes with a horizon of 25 steps, MPE environments \cite{lowe2017multi} for 30k episodes (horizon 25), and Melting Pot (Harvest) \cite{leibo2021scalable} for 10k episodes (horizon 500). To ensure statistical robustness, all results for MPE and Melting Pot are reported as mean $\pm$ standard deviations over 5 independent runs.

\begin{figure}[t]
    \centering
    \includegraphics[width=0.495\columnwidth]{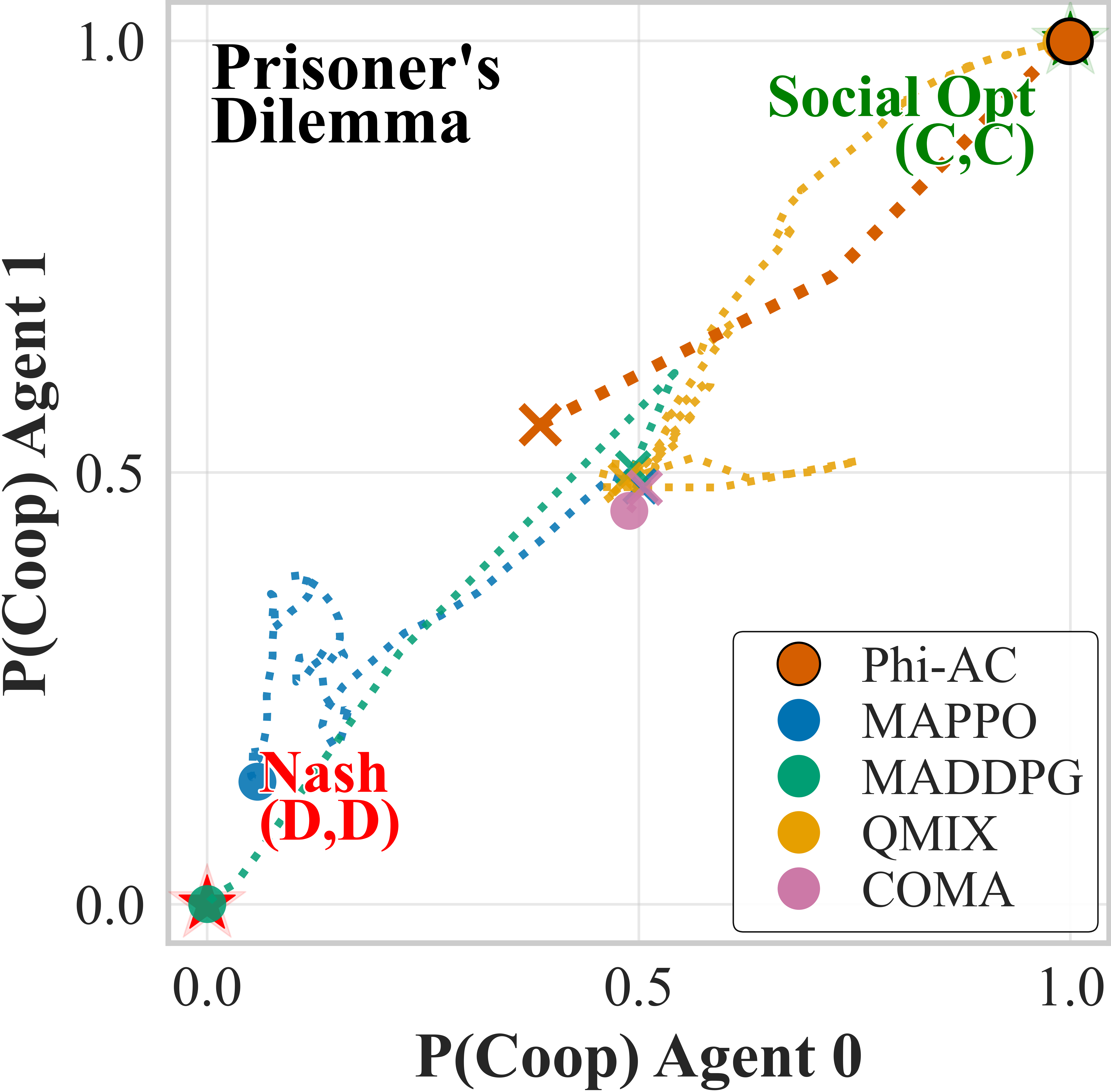}%
    \hfill%
    \includegraphics[width=0.495\columnwidth]{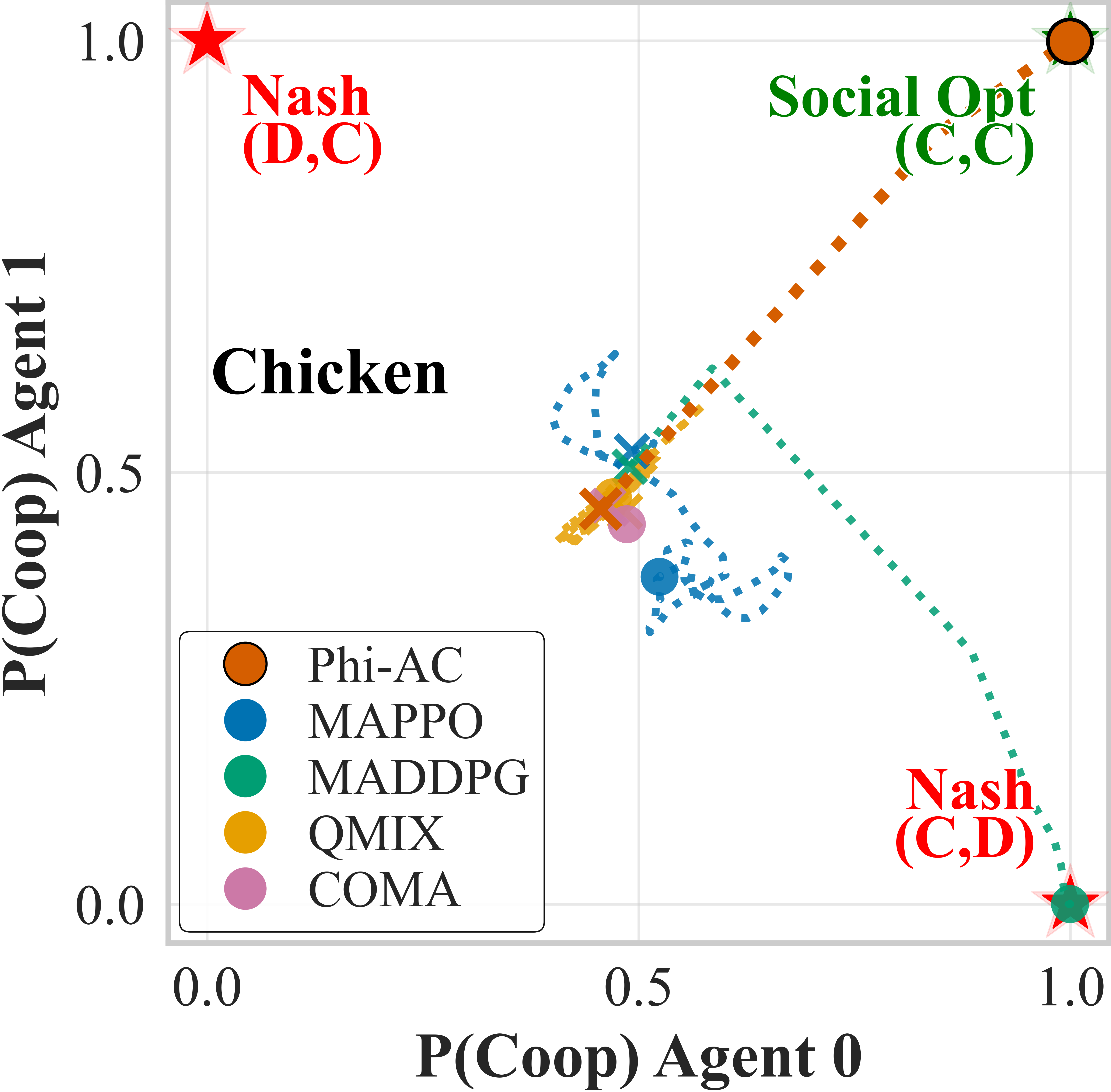}
    
    \includegraphics[width=0.495\columnwidth]{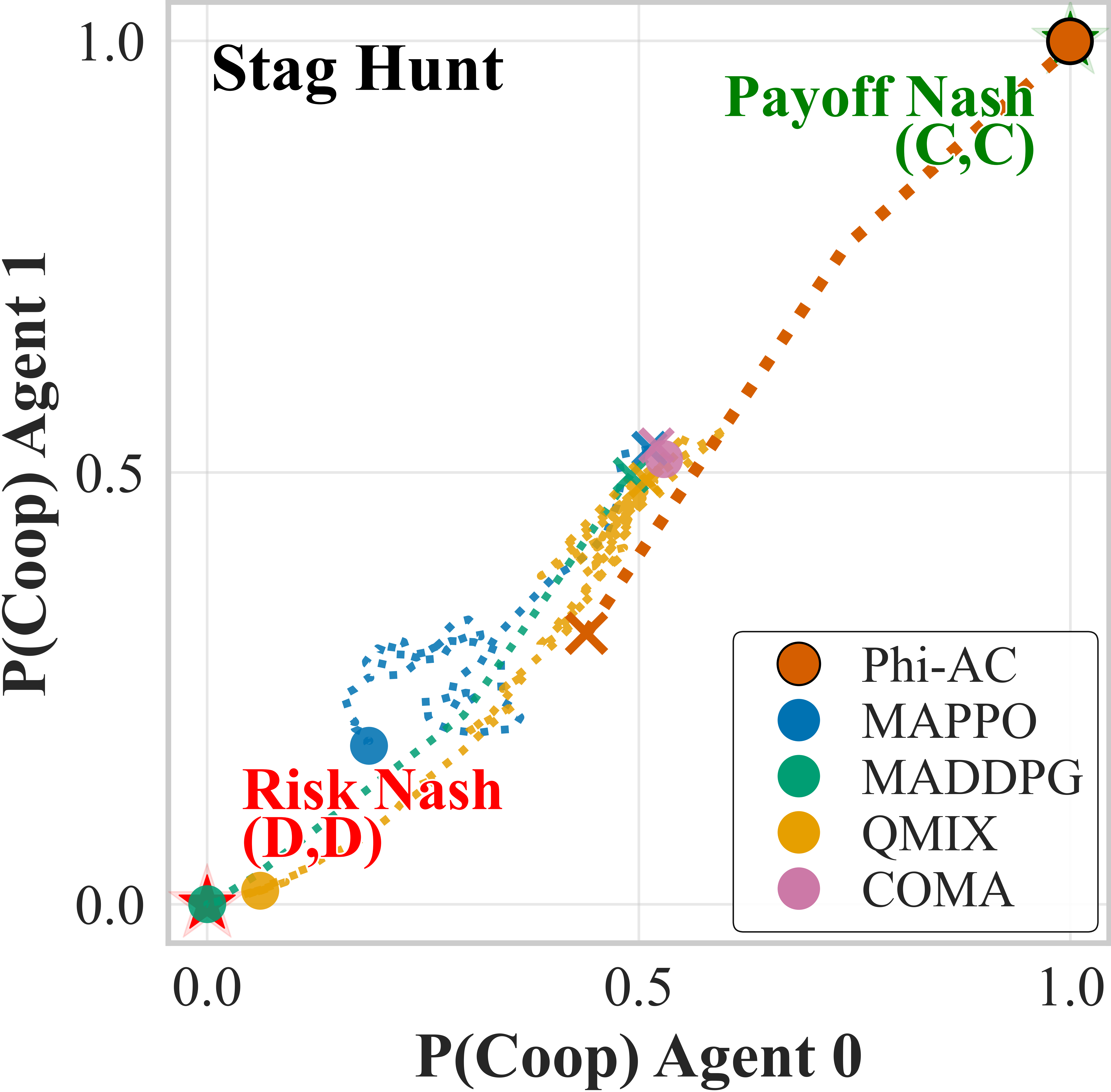}
    
    \caption{Policy Trajectories. (Top-L) Prisoner's Dilemma, (Top-R) Chicken, (Bot) Stag Hunt. $\Phi$-AC (Red) consistently reaches the intended high-welfare equilibria.}
    \label{fig:matrix_trajectories}
\end{figure}

\subsection{Diagnostic Analysis: Iterated Matrix Games}
\label{sec:matrix_games}

We first employ Iterated Matrix Games to empirically verify the equilibrium convergence of each model. We consider a two-player setting where agents interact repeatedly, allowing them to adapt strategies based on history. We select three canonical games, each representing a distinct coordination failure mode. Fig.~\ref{fig:matrix_trajectories} visualizes the learning trajectories.

\paragraph{Prisoner's Dilemma.}
This game tests whether agents can overcome the temptation to defect and coordinate on mutual cooperation. 
As shown in Fig.~\ref{fig:matrix_trajectories} (Top-L), MADDPG and MAPPO tend to converge toward the inefficient NE (D,D), while QMIX and $\Phi$-AC reach the social optimum (C,C). Unlike QMIX, which benefits from the monotonic payoff structure in this setting, $\Phi$-AC reaches mutual cooperation through regret-aware coordination.

\paragraph{Chicken.}
This game tests coordination under two asymmetric equilibria, where either agent can benefit if the other yields. 
As shown in Fig.~\ref{fig:matrix_trajectories} (Top-R), baselines often drift toward one-sided or unstable outcomes, whereas $\Phi$-AC promotes balanced coordination by reducing regret disparities between agents.

\paragraph{Stag Hunt.}
This game tests whether agents can coordinate on a high-payoff action that is beneficial only when both agents choose it, rather than falling back to the safer but lower-reward option. 
As shown in Fig.~\ref{fig:matrix_trajectories} (Bottom), baselines tend to favor the conservative risk-dominant outcome, while $\Phi$-AC consistently reaches the payoff-dominant cooperative equilibrium.

\begin{figure}[t]
    \centering
{\footnotesize
\begin{tabular}{@{}ccc@{}}
    \textcolor{maddpgcolor}{\rule{1.0em}{2.0pt}} MADDPG &
    \textcolor{mappocolor}{\rule{1.0em}{2.0pt}} MAPPO &
    \textcolor{qmixcolor}{\rule{1.0em}{2.0pt}} QMIX \\
    \multicolumn{3}{c}{
        \textcolor{comacolor}{\rule{1.0em}{2.0pt}} COMA
        \quad
        \textcolor{phiaccolor}{\rule{1.0em}{2.0pt}} $\Phi$-AC
}
\end{tabular}
}
    
    \begin{subfigure}[t]{0.49\linewidth}
        \centering
        \includegraphics[width=\linewidth]{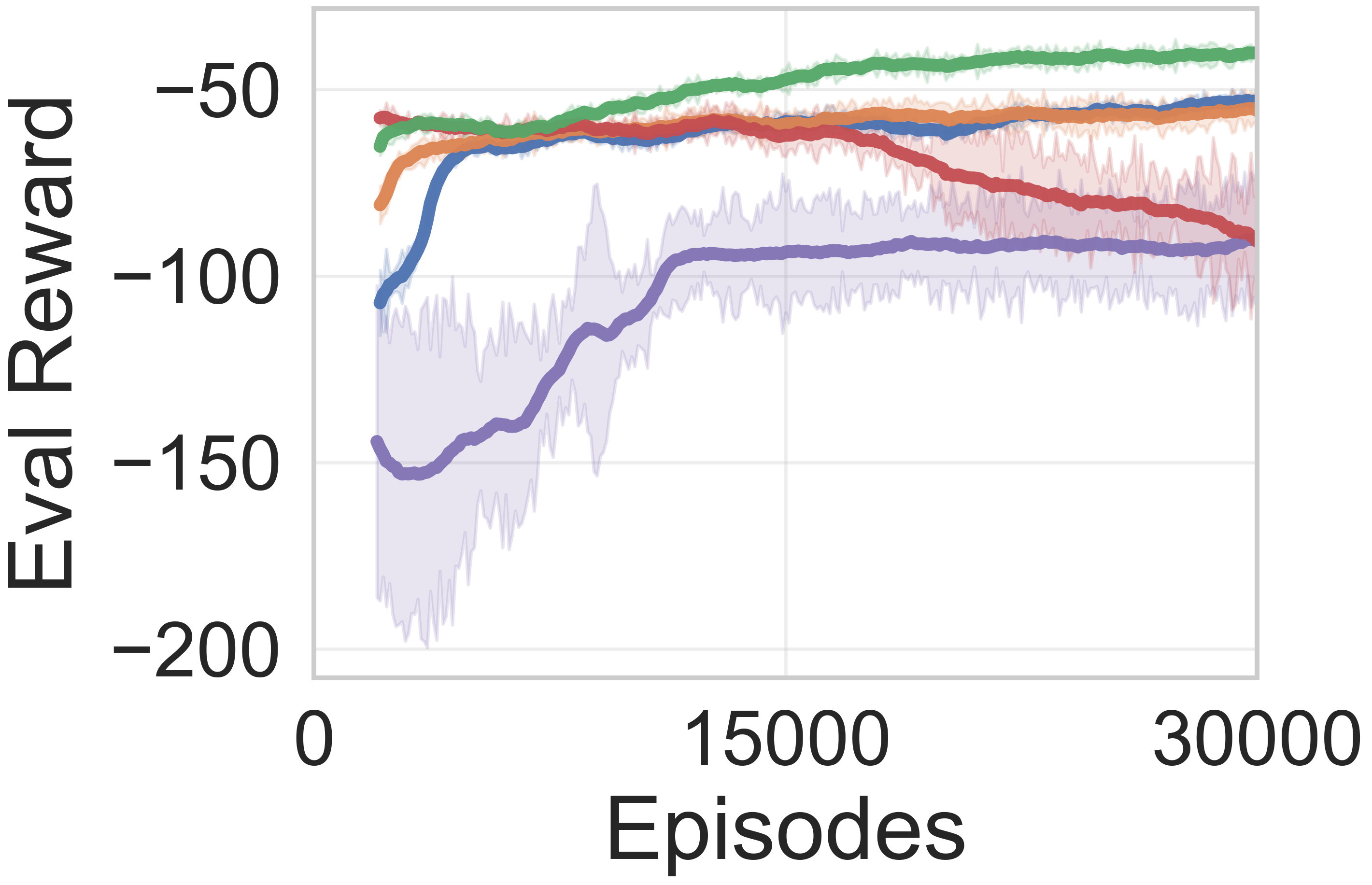}
        \caption{Coop: Reward}
        \label{fig:coop_reward}
    \end{subfigure}
    \hfill
    \begin{subfigure}[t]{0.49\linewidth}
        \centering
        \includegraphics[width=\linewidth]{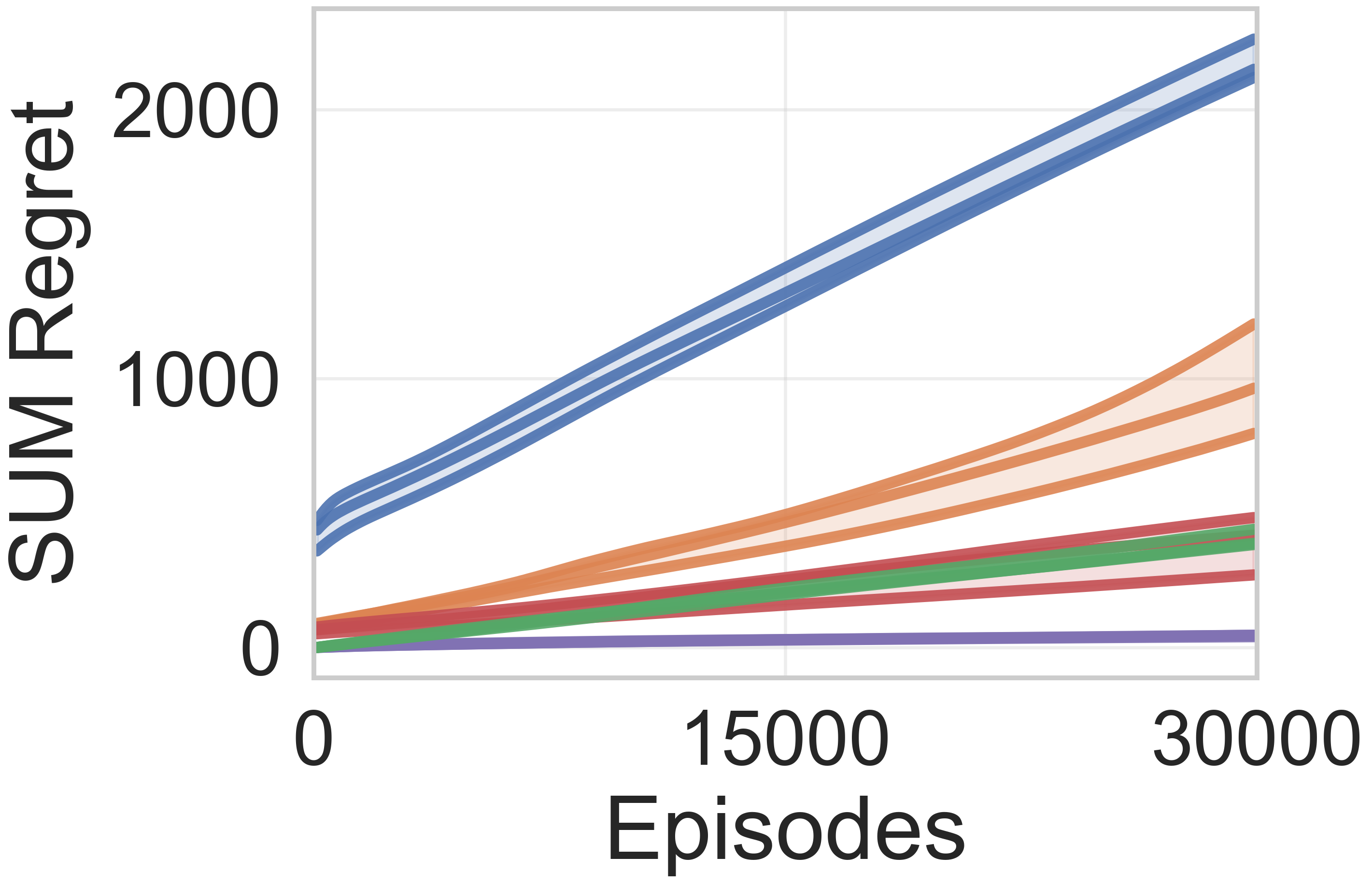}
        \caption{Coop: Regret}
        \label{fig:coop_regret}
    \end{subfigure}

    \begin{subfigure}[t]{0.49\linewidth}
        \centering
        \includegraphics[width=\linewidth]{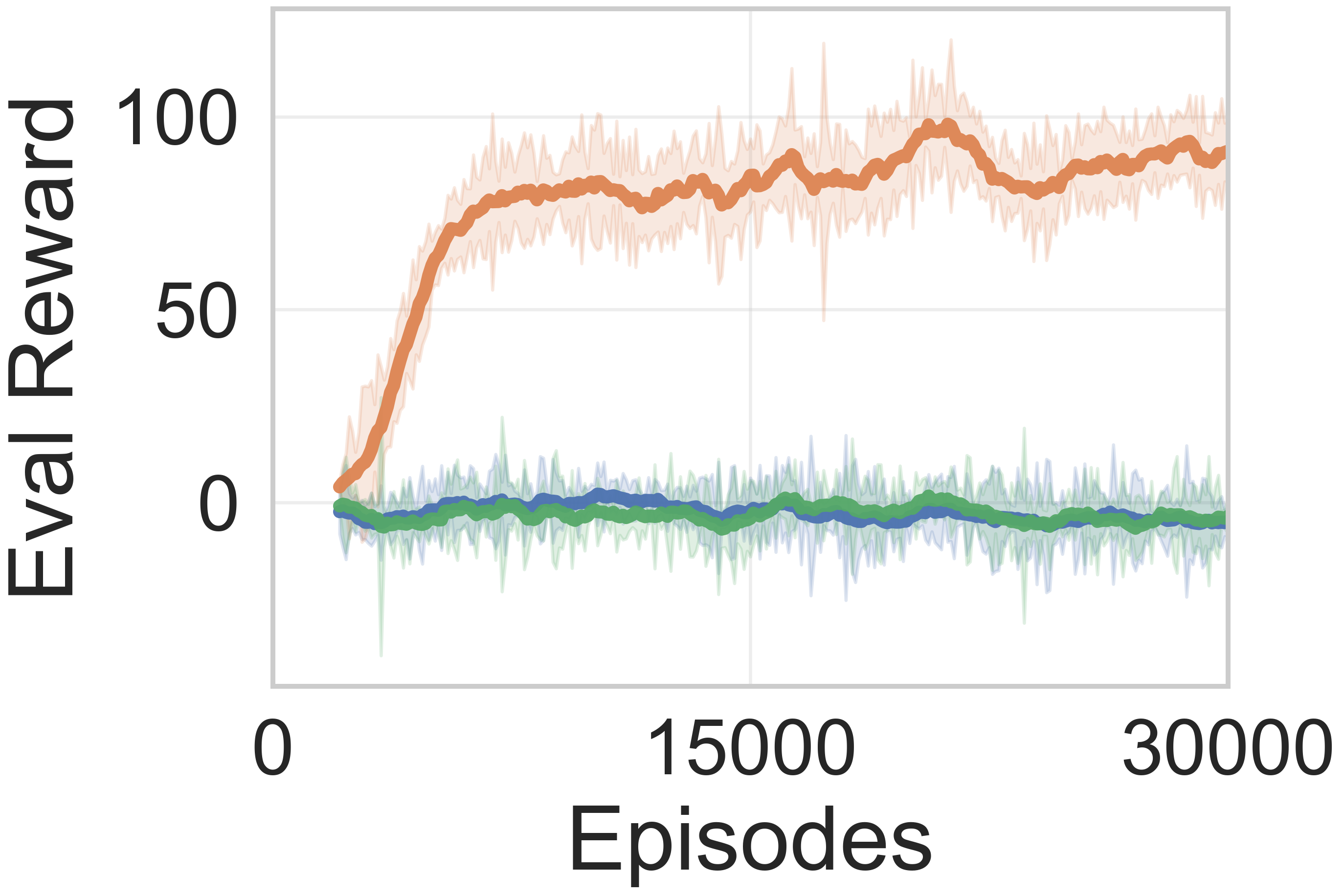}
        \caption{Zero-Sum: Reward}
        \label{fig:adv_reward}
    \end{subfigure}
    \hfill
    \begin{subfigure}[t]{0.49\linewidth}
        \centering
        \includegraphics[width=\linewidth]{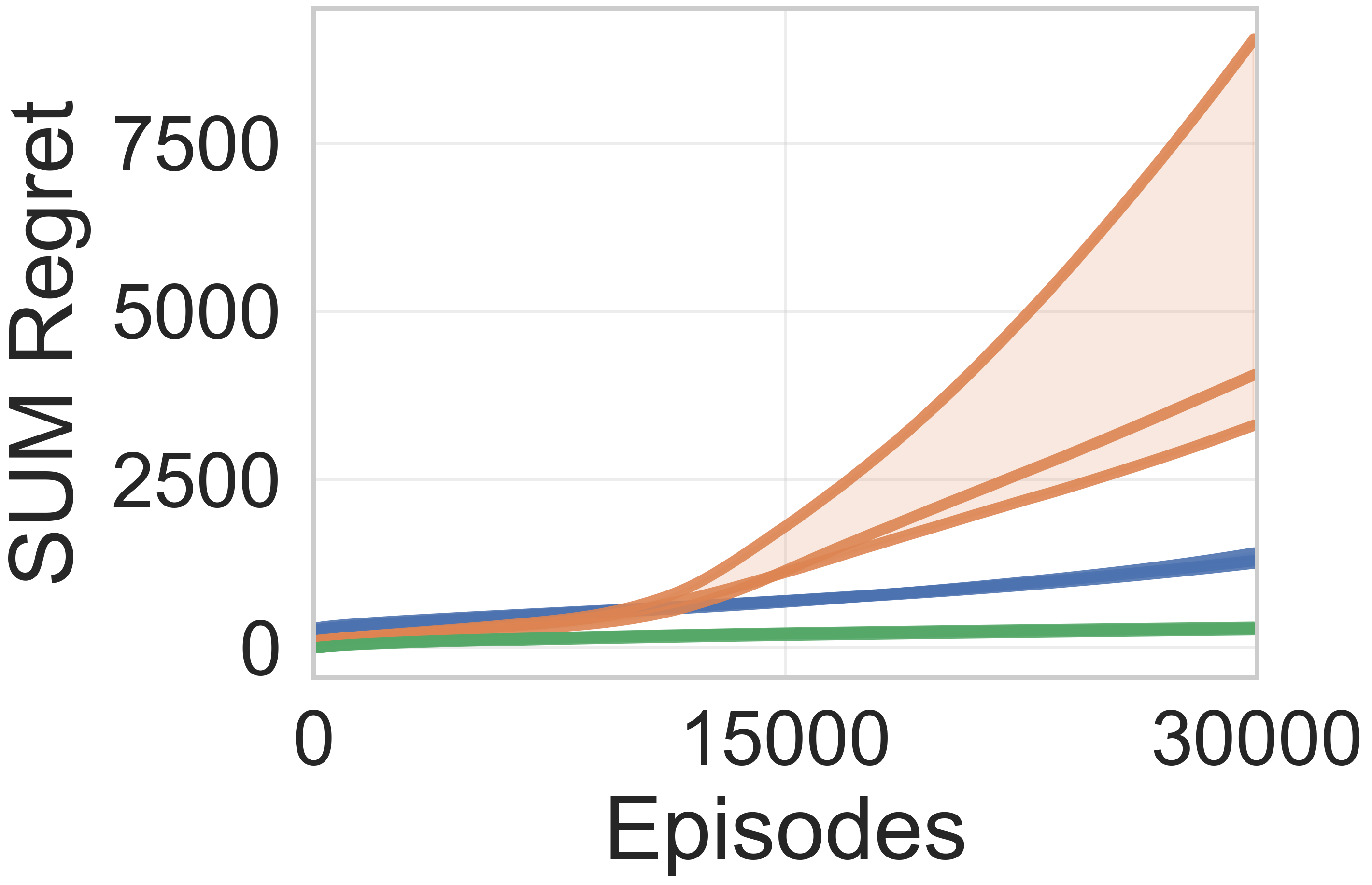}
        \caption{Zero-Sum: Regret}
        \label{fig:adv_regret}
    \end{subfigure}

    \begin{subfigure}[t]{0.49\linewidth}
        \centering
        \includegraphics[width=\linewidth]{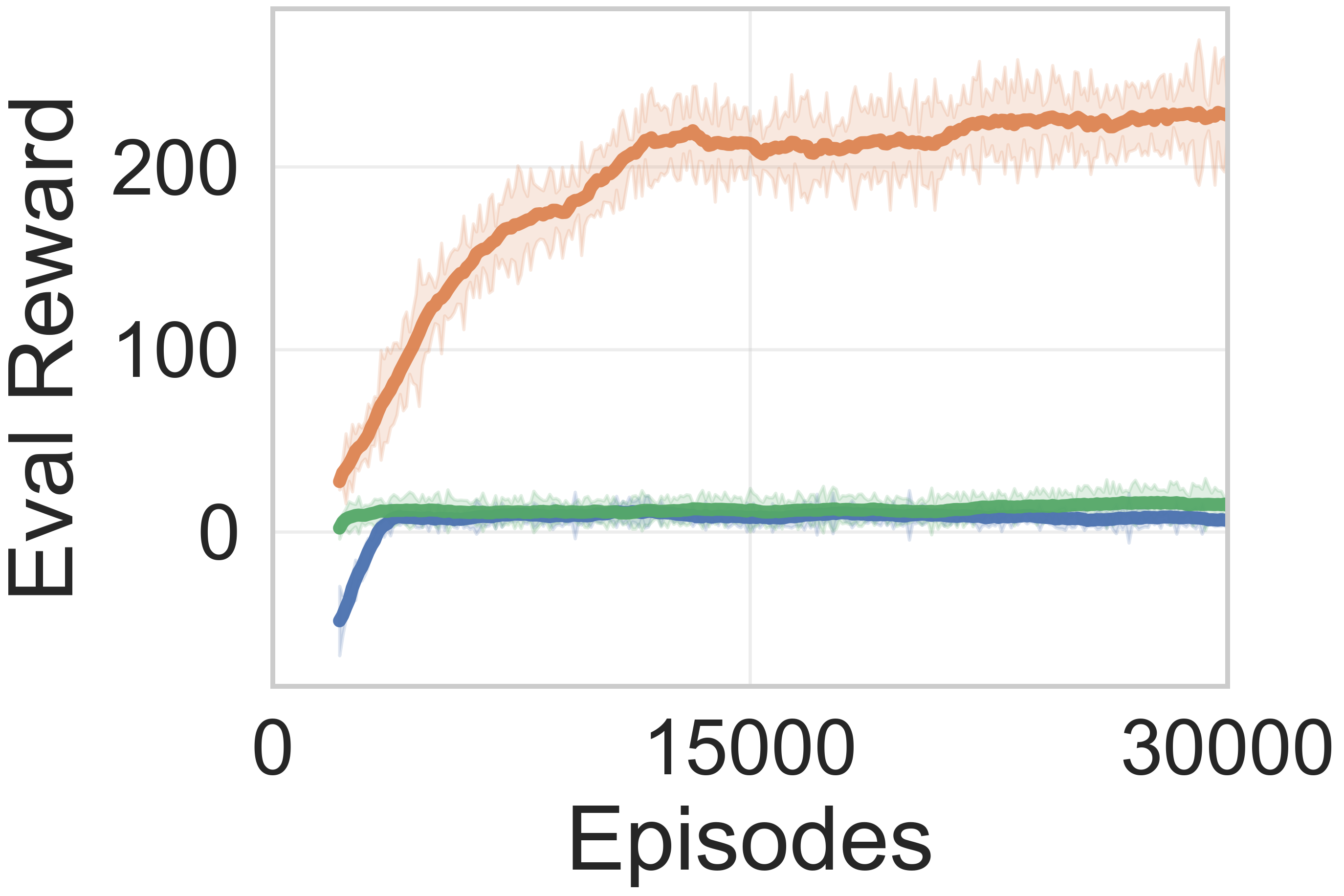}
        \caption{Mixed: Reward}
        \label{fig:mixed_reward}
    \end{subfigure}
    \hfill
    \begin{subfigure}[t]{0.49\linewidth}
        \centering
        \includegraphics[width=\linewidth]{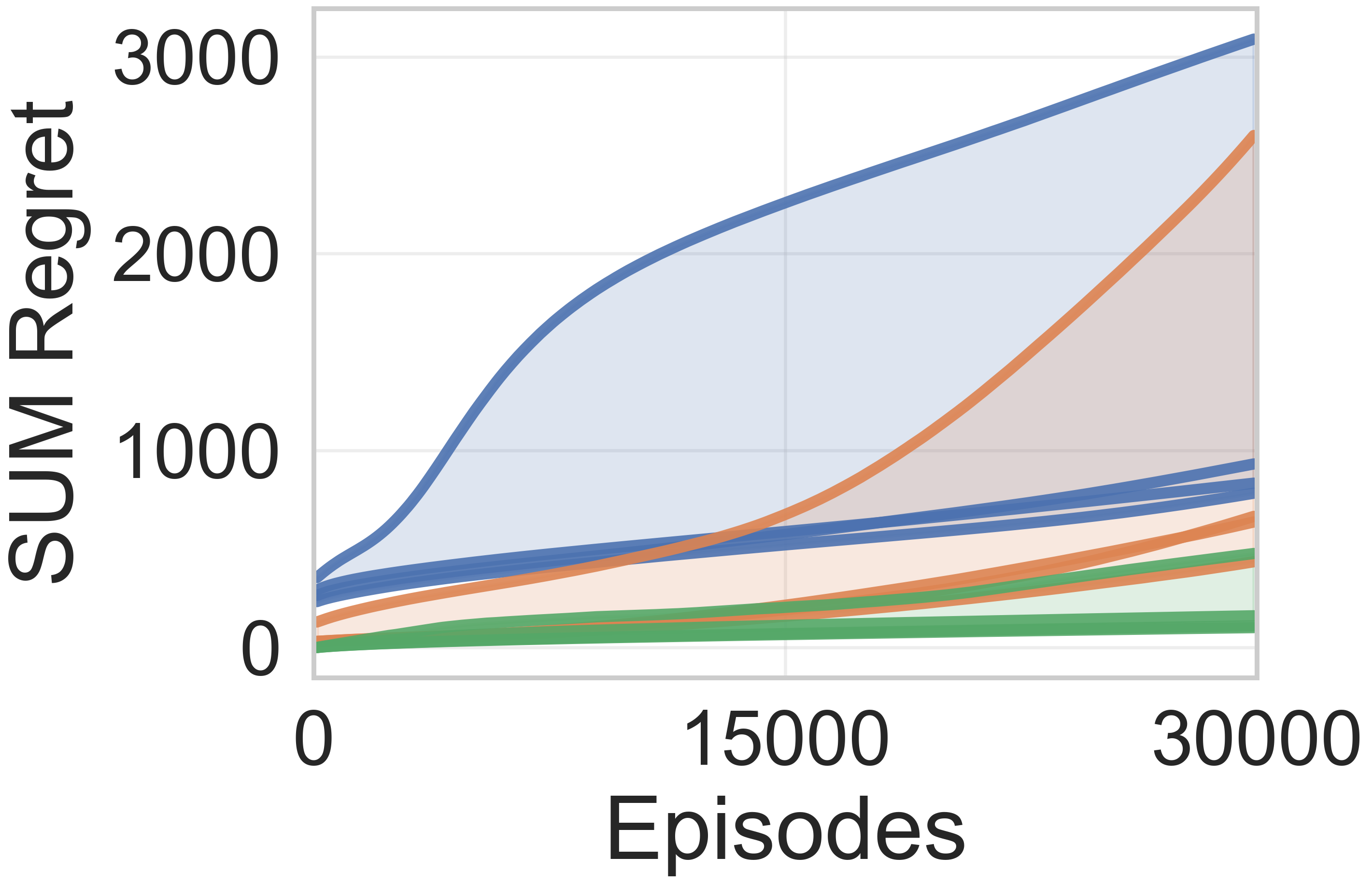}
        \caption{Mixed: Regret}
        \label{fig:mixed_regret}
    \end{subfigure}
    \caption{Training Dynamics. Rows represent environments: Cooperative, Zero-Sum, and Mixed. Left panels show evaluation rewards, and right panels show cumulative regret.}
    \label{fig:all_results_dense}
\end{figure}

\subsection{Selection under Scale: MPE}
\label{sec:mpe_experiments}

We extend our evaluation to the MPE to verify whether the selection capability scales to high-dimensional state-action spaces via our centralized attention critic. We report Mean Reward (efficiency), Cumulative Regret (Convergence), and Regret Gap ($|\mathcal{R}^{\text{cum}}_i - \mathcal{R}^{\text{cum}}_j|$, fairness). Table~\ref{tab:final_results} summarizes the quantitative results, and Figure~\ref{fig:all_results_dense} visualizes the training dynamics.

\paragraph{Cooperative: Avoiding Inefficient Equilibria.}
In \texttt{simple\_spread}, agents must cover landmarks while avoiding collisions. COMA and QMIX achieve low regret but converge to lower-welfare solutions, suggesting that stability alone does not imply efficient equilibrium selection. In contrast, $\Phi$-AC achieves the highest reward (-42.24). By balancing welfare maximization with regret constraints, it avoids premature stabilization and reaches a higher-welfare configuration.

\paragraph{Zero-Sum: Stability vs. Exploitation.}
High individual rewards in \texttt{simple\_adversary} often signal exploitation rather than equilibrium. MAPPO achieves a high reward (91.17) but suffers from large regret gap (5926.82), indicating a non-stationary cycle where one agent continually exploits a weak opponent without convergence. $\Phi$-AC converges to a reward near zero (-3.27) with significantly lower regret, suggesting a more stable low-regret regime in which unilateral exploitation incentives are reduced.

\paragraph{Mixed-Motive: Scalable Fairness.}
\texttt{simple\_tag} involves conflicting goals within a general-sum framework. Baselines exhibit large regret gaps ($>2000$), indicating highly imbalanced deviation incentives consistent with winner-takes-all behavior. $\Phi$-AC reduces this gap by nearly two orders of magnitude (53.79). This suggests that the RB-SWO mechanism can scale to high-dimensional domains, enforcing equity by ensuring that the dissatisfaction (regret) is distributed equitably among agents.

\begin{table}[h]
\centering
    \small 
    \renewcommand{\arraystretch}{0.9}
    \setlength{\tabcolsep}{3pt}
    
    \resizebox{\columnwidth}{!}{%
    \begin{tabular}{lccc}
        \toprule
        \textbf{Model} & \textbf{Reward} & \textbf{Cum. Regret} & \textbf{Regret Gap} \\
        \midrule
        \multicolumn{4}{l}{\textit{\textbf{Cooperative (Spread)}}} \\ 
        MADDPG & $-53.55 \pm 0.55$  & $2177.97 \pm 46.98 $ & $165.32 \pm 61.98$  \\
        MAPPO  & $-55.68 \pm 3.11$  & $988.00 \pm 82.29$   & $527.92 \pm 272.01$ \\
        COMA   & $-91.24 \pm 13.08$ & \textbf{43.72 $\pm$ 11.54}    & \textbf{39.93 $\pm$ 25.37}    \\
        QMIX   & $-88.49 \pm 11.54$ & $390.13 \pm 41.66$   & $248.44 \pm 60.42$  \\
        $\Phi$-AC & \textbf{-42.24 $\pm$ 0.21}  & $460.46 \pm 34.69$    & $260.64 \pm 119.98$ \\
        
        \midrule
        \multicolumn{4}{l}{\textit{\textbf{Zero-Sum (Adversary)}}} \\
        MADDPG & $-4.91 \pm 3.65$    & $1316.86 \pm 105.13$ & $187.11 \pm 122.92$ \\
        MAPPO  & $91.17 \pm 2.96$    & $5470.50 \pm 1845.02$ & $5926.82 \pm 6044.88$ \\
        $\Phi$-AC & \textbf{-3.27 $\pm$ 1.63} & \textbf{281.30 $\pm$ 11.10}    & \textbf{43.32 $\pm$ 18.28}    \\
        
        \midrule
        \multicolumn{4}{l}{\textit{\textbf{Mixed-Motive (Tag)}}} \\
        MADDPG & $6.83 \pm 0.93$     & $1409.98 \pm 50.54$  & $2321.88 \pm 140.22$ \\
        MAPPO  & $227.69 \pm 13.24$  & $1084.68 \pm 202.97$ & $2192.30 \pm 654.54$ \\
        $\Phi$-AC & \textbf{12.55 $\pm$ 7.32}  & \textbf{481.56 $\pm$ 20.51}    & \textbf{53.79 $\pm$ 12.09}    \\
        \bottomrule
    \end{tabular}%
    }
    \caption{Scalability results on MPE. $\Phi$-AC achieves the best balance of efficiency and stability. Note: In zero-sum settings, a reward near 0 with low regret implies a stable NE, whereas high rewards indicate exploitation.}
    \label{tab:final_results}
\end{table}

\subsection{Ablation on Learning Objectives}
\label{sec:ablation_study}

We dissect the contribution of our proposed learning objectives in \texttt{simple\_tag}. The results, summarized in Table~\ref{tab:ablation_results}, reveal that each component is critical for preventing unstable learning outcomes.

\begin{enumerate}
    \item \textbf{w/o Regret Matching (RM) Bonus:} 
    Removing the regret-bias term ($\beta = 0$) leads to unstable learning dynamics. The total regret increases substantially to \textbf{1933.04} with an extreme standard deviation ($\pm$ 2007.34). This indicates that without the RM bonus, the learning dynamics fail to converge, oscillating severely or diverging completely.
    
    \item \textbf{w/o Fairness Objective (RB-SWO):} 
    Removing the fairness constraint ($\alpha^{\text{fair}} \to 0$) degrades both stability ($481.56 \to 510.11$ in Regret) and efficiency. This suggests that the fairness objective acts as a dual-purpose regularizer: it not only prevents inefficient outcomes (``Coordination Collapse") but also accelerates convergence by constraining the search space to the stable region.
\end{enumerate}

\begin{table}[h]
\centering
\small
\resizebox{\columnwidth}{!}{%
\begin{tabular}{lcc}
\toprule
\textbf{Model Variant} & \textbf{Final Reward} & \textbf{Total Regret} \\
\midrule
\textbf{Full $\Phi$-AC (Proposed)} & \textbf{12.55 $\pm$ 7.32} & \textbf{481.56 $\pm$ 20.51} \\
\midrule
w/o RM Bonus & 12.53 $\pm$ 7.95 & 1933.04 $\pm$ 2007.34 \\
w/o Fairness Objective & 10.94 $\pm$ 7.19 & 510.11 $\pm$ 28.46 \\
\bottomrule
\end{tabular}%
}
\caption{Ablation results on \texttt{simple\_tag}. Removing the RM Bonus leads to severe instability (large regret variance), while removing Fairness degrades both convergence and reward. $\Phi$-AC achieves the most stable and efficient equilibrium.}
\label{tab:ablation_results}
\end{table}

\subsection{Robustness in Dilemmas: Melting Pot}
\label{sec:melting_pot}
Finally, we evaluate $\Phi$-AC on Harvest, an SSD that serves as a challenging benchmark for sustainable coordination under shared resource constraints. In this environment, apples regenerate based on the density of nearby unharvested apples, creating a ``Tragedy of the Commons'' where individually greedy behavior can undermine long-term collective productivity. We compare $\Phi$-AC against MAPPO (centralized policy-gradient baseline) and QMIX (value decomposition baseline) using the Sustainability Index (SI), defined as
$
SI = (1 - \text{Gini}) \times \sum_i \text{Return}_i,
$
which captures the trade-off between productivity and fairness.

\begin{figure}[t]
    \centering
    \begin{subfigure}[t]{0.485\columnwidth}
        \centering
        \includegraphics[width=\linewidth]{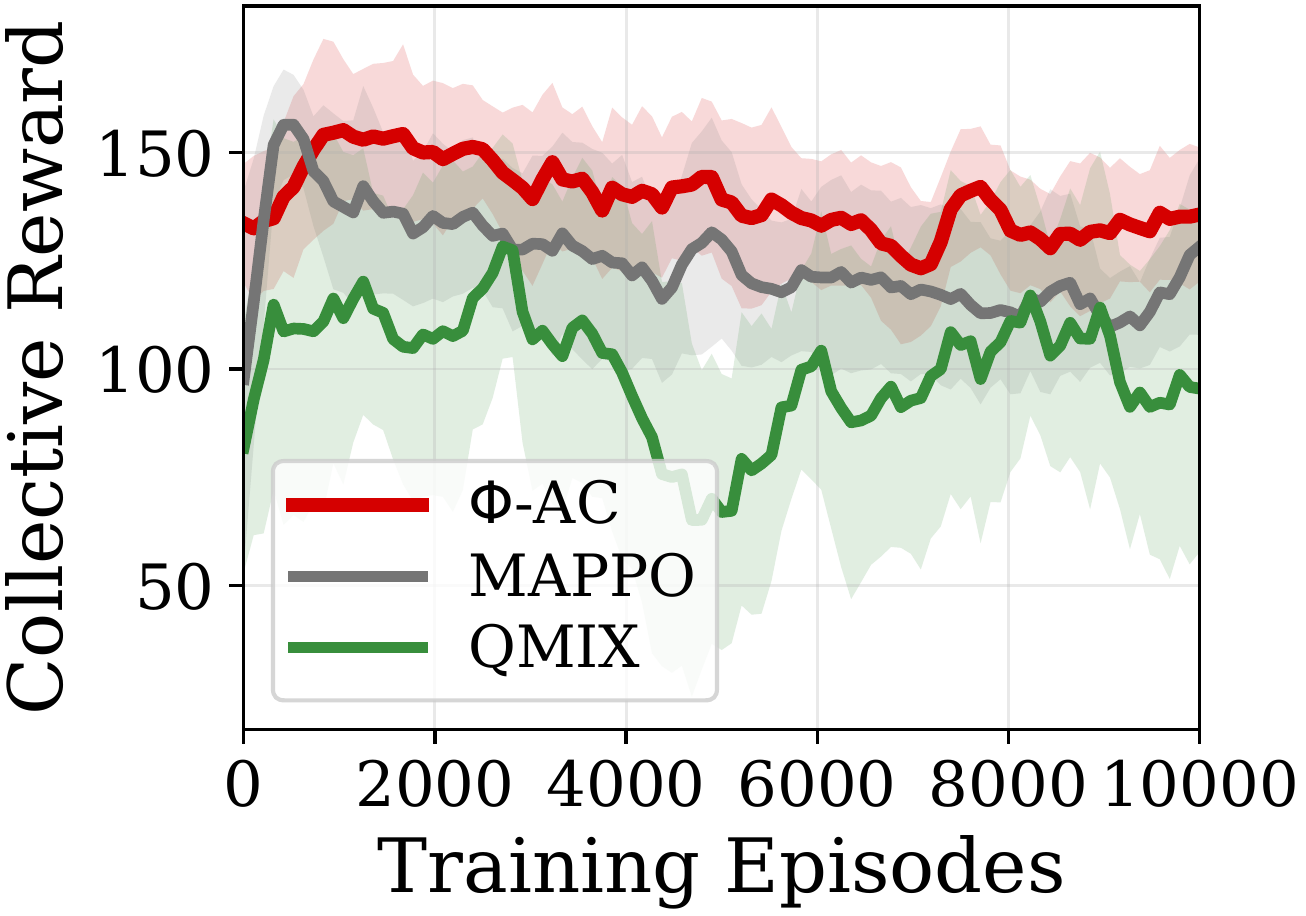}
        \caption{Collective Reward}
        \label{fig:harvest_reward}
    \end{subfigure}
    \hfill
    \begin{subfigure}[t]{0.485\columnwidth}
        \centering
        \includegraphics[width=\linewidth]{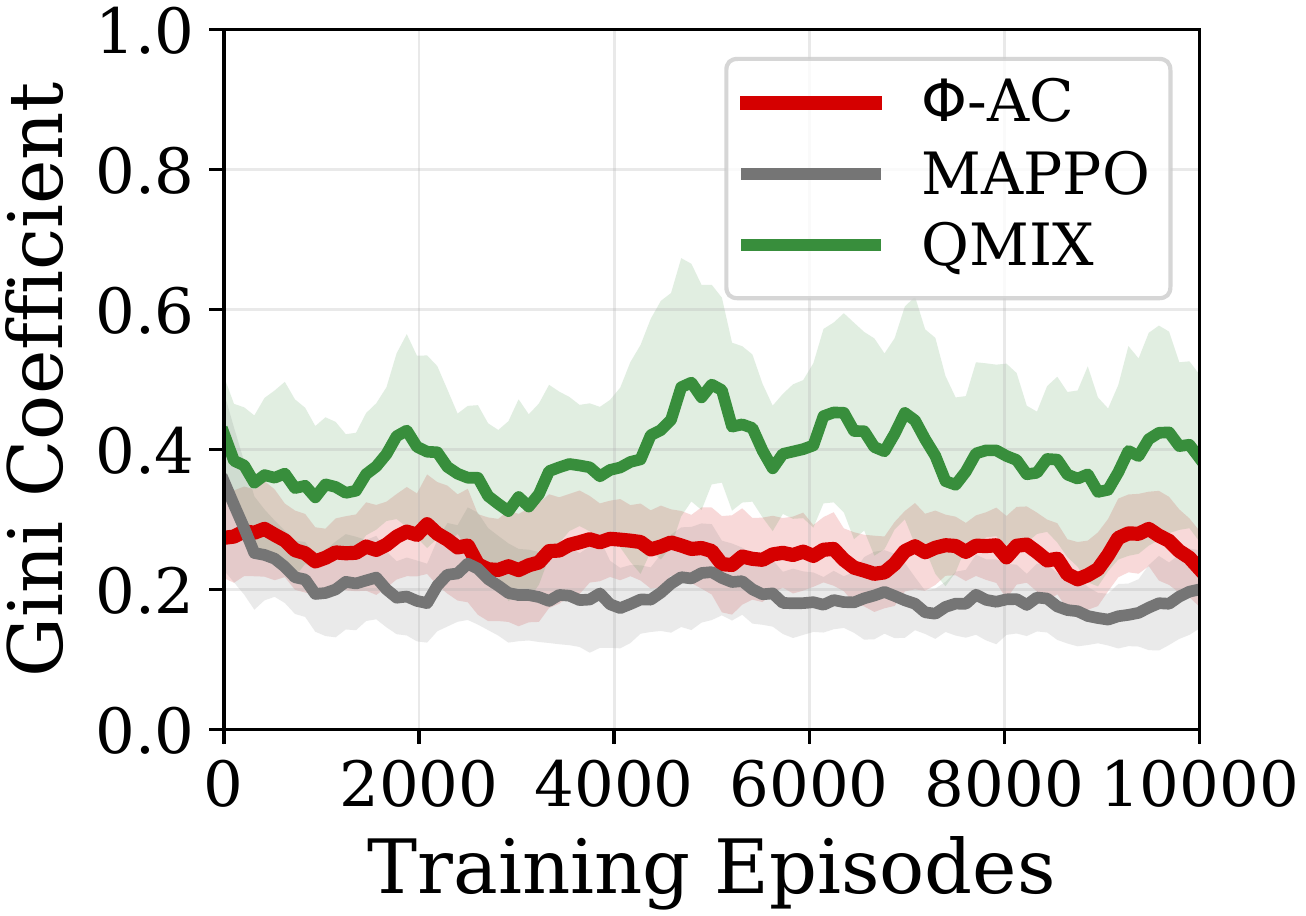}
        \caption{Gini Coefficient}
        \label{fig:harvest_gini}
    \end{subfigure}

    \begin{subfigure}[t]{0.485\columnwidth} 
        \centering
        \includegraphics[width=\linewidth]{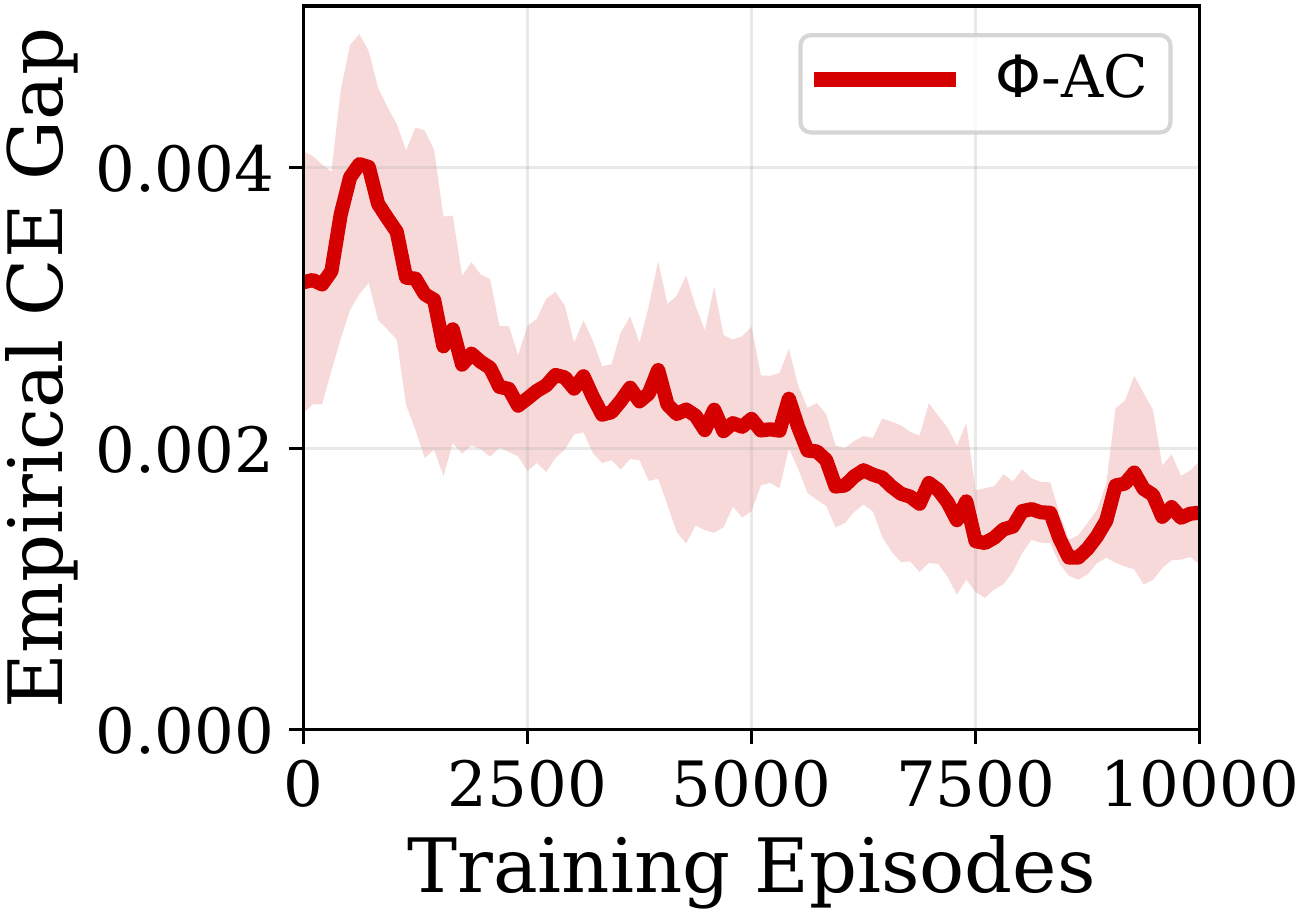}
        \caption{Empirical CE Gap}
        \label{fig:harvest_ce_gap}
    \end{subfigure}
    \hfill
    \begin{subfigure}[t]{0.49\columnwidth}
        \centering
        \includegraphics[width=\linewidth]
        {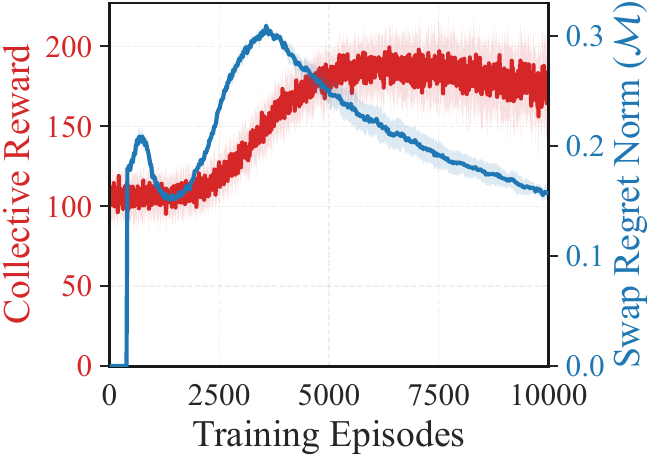}
        \caption{Reward-Regret Dynamics}
        \label{fig:harvest_mechanism}
    \end{subfigure}
    
    \caption{Dynamics in Harvest. 
    (a,b) $\Phi$-AC achieves high collective reward while maintaining competitive fairness. 
    (c) The empirical CE-gap proxy remains low during evaluation. 
    (d) Training reward-regret dynamics show decreasing regret alongside sustained collective reward.}
    \label{fig:harvest_results}

\end{figure}

\paragraph{Escaping Resource Collapse (Fig.~\ref{fig:harvest_results} (a,b)).}
As shown in Table~\ref{tab:harvest_results}, $\Phi$-AC achieves the highest collective reward (134.33) while maintaining competitive fairness relative to MAPPO. In contrast, QMIX suffers from severe inequality and substantially lower collective productivity, resulting in the weakest Sustainability Index. Although MAPPO achieves the lowest Gini coefficient, its lower collective reward leads to a smaller overall sustainability score than $\Phi$-AC. These results suggest that regret-based equilibrium selection helps maintain productive coordination without collapsing into resource-depleting behavior.

\paragraph{Reward--Regret Dynamics (Fig.~\ref{fig:harvest_results} (c,d)).}
To further analyze the learned coordination dynamics, we additionally measure the empirical CE gap, defined as the maximum positive swap regret estimated by the critic during evaluation. As shown in Fig.~\ref{fig:harvest_results}(c), the empirical CE gap of $\Phi$-AC remains consistently low throughout training, suggesting reduced unilateral deviation incentives among agents. 

This interpretation is further supported by the reward-regret dynamics in Fig.~\ref{fig:harvest_results}(d). During early training, swap regret temporarily increases as agents explore alternative coordination strategies. As learning progresses, however, the regret signal decreases while collective reward remains high. This indicates that agents gradually settle into a stable high-welfare coordination pattern rather than maximizing reward through unstable exploitative behavior.

\begin{table}[h]
\centering
\small
\resizebox{\columnwidth}{!}{%
\begin{tabular}{lccc}
\toprule
Model & Coll. Reward & Gini Coeff & Sust. Index \\
\midrule
$\Phi$-AC (Ours) 
& \textbf{134.33 $\pm$ 3.46} 
& 0.26 $\pm$ 0.03 
& \textbf{99.41 $\pm$ 2.67} \\

MAPPO 
& 116.55 $\pm$ 7.96 
& \textbf{0.18 $\pm$ 0.03} 
& 95.48 $\pm$ 7.60 \\

QMIX 
& 94.49 $\pm$ 25.08 
& 0.40 $\pm$ 0.05 
& 59.27 $\pm$ 18.27 \\
\bottomrule
\end{tabular}%
}
\caption{Performance summary in harvest (Evaluation, Last 10\%). $\Phi$-AC achieves the highest Sustainability Index while maintaining strong collective productivity.}
\label{tab:harvest_results}
\end{table}

\section{Conclusion}
\label{sec:conclusion}

In this work, we introduced $\Phi$-AC, a principled framework designed to integrate game-theoretic swap regret minimization with deep MARL. By shifting the paradigm from reward-only optimization to equilibrium selection, $\Phi$-AC addresses the limitations of conventional MARL objectives in distinguishing efficient coordination from suboptimal stability.

Our empirical evaluation supports $\Phi$-AC across a hierarchy of complexity.
At the diagnostic level (Matrix Games), we visualized robust steering toward social optima. 
At the scalable level (MPE), we demonstrated effective fairness-aware coordination. Most significantly, in the complex SSD (Harvest), we provided empirical evidence of high-welfare coordination behavior consistent with correlated-equilibrium principles.

The observed decrease in smoothed swap regret alongside increasing collective reward suggests that $\Phi$-AC promotes sustained coordination rather than short-term reward maximization.
A current limitation is its reliance on discrete action spaces. Future work will extend regret estimation to continuous control, broadening game-theoretic equilibrium selection to physical decision-making domains.

\clearpage
\bibliographystyle{named}
\bibliography{ijcai26}

\clearpage
\appendix


\section{Theoretical Proofs and Background}
\label{app:theory}

In this section, we provide the theoretical foundations connecting swap regret minimization to CE and outline the proof for the convergence proposition presented in the main text.

\subsection{Connection between Swap Regret and CE}
As established in foundational algorithmic game theory \cite{hart2000simple,blum2007external}, there is a fundamental equivalence between no-swap-regret dynamics and the set of CE.
Recall the definition of the swap-regret vector $\boldsymbol{\mathcal{R}}_{i,T}^{\Phi} \in \mathbb{R}^{|\Phi_i|}$, whose components are indexed by deviation mappings $\phi \in \Phi_i$. 
Let $t \in \{1, \dots, T\}$ denote the discrete time step, $a_{i,t} \in \mathcal{A}_i$ be the action taken by agent $i$, and $\mathbf{a}_{-i,t}$ be the joint action of all other agents. For each deviation $\phi$, the average swap regret over $T$ steps is evaluated using the reward function $r_i$:
\begin{equation}
    \mathcal{R}_{i,T}^{\Phi}(\phi)
    =
    \frac{1}{T}\sum_{t=1}^{T}
    \left[
    r_i(\phi(a_{i,t}),\mathbf{a}_{-i,t})
    -
    r_i(a_{i,t},\mathbf{a}_{-i,t})
    \right].
\end{equation}
A distribution $z \in \Delta(\mathcal{A})$ over the joint action space $\mathcal{A}$ constitutes a CE if no agent can improve their expected utility by unilaterally applying any swap $\phi$. Driving the maximum positive regret, denoted as $\max{\phi \in \Phi_i} [\mathcal{R}{i,T}^{\Phi}(\phi)]+$ (where $[\cdot]_+$ indicates the positive part operator), to zero for every agent implies that the empirical distribution of play approaches the CE set \cite{hart2000simple}.

\paragraph{Intuition of Regret Matching.}
To provide deeper intuition for Regret Matching (RM), it is instructive to consider the rationale behind increasing the probability of actions with high positive regret. 
Intuitively, a positive swap regret $\mathcal{R}(a \to a') > 0$ represents a "missed opportunity": the additional utility the agent \textit{would have gained} had they played $a'$ instead of $a$. 
By strictly following RM dynamics (or its smooth approximation via Softmax), the agent probabilistically corrects these past errors, steering the joint policy toward a state where no such missed opportunities exist (i.e., equilibrium).

\paragraph{From Normal-Form to Sequential Games.}
While classical swap regret is defined for normal-form games, extending it to Markov games requires evaluating deviations over entire trajectories. Motivated by the one-shot deviation principle, we approximate trajectory-level deviation incentives through local Markovian deviations evaluated by the stationary action-value function $Q_i^\pi(s,\mathbf a)$. This approximation is most naturally interpreted under stationary Markov policies and Markovian deviation classes. Furthermore, while exact theoretical regret relies on the arithmetic time-average, in our implementation, we use an Exponential Moving Average (EMA) as a practical approximation to stabilize regret estimates in non-stationary deep MARL. Unlike the arithmetic average used in the classical convergence proof, EMA emphasizes recent deviation incentives and should be interpreted as an implementation-level stabilization mechanism rather than a separate convergence guarantee.

\subsection{Proof Sketch of Proposition 1}
\textbf{Proposition 1.} \textit{Let the critic's regret prediction error be bounded by $\|\hat{\boldsymbol{\mathcal{R}}}^\psi - \boldsymbol{\mathcal{R}}^{\text{true}}\|_\infty \leq \epsilon_c$. 
If agents follow smooth regret matching dynamics using $\hat{\boldsymbol{\mathcal{R}}}^\psi$, the joint empirical distribution converges to the set of $\epsilon$-CE, where $\epsilon = \mathcal{O}(\epsilon_c + \delta_{\text{regret}} + \tau \log |\mathcal{A}_i|)$.}

\begin{proof}
Let $\boldsymbol{\mathcal{R}}^{\text{true}}_t$ be the true regret vector at time $t$, and $\hat{\boldsymbol{\mathcal{R}}}^\psi_t$ be the critic's estimate such that $\|\hat{\boldsymbol{\mathcal{R}}}^\psi_t - \boldsymbol{\mathcal{R}}^{\text{true}}_t\|_\infty \le \epsilon_c$.
Standard no-regret algorithms guarantee that the time-averaged regret vanishes at a rate of $\mathcal{O}(1/\sqrt{T})$ using the true regret values.
When using the approximate regret $\hat{\boldsymbol{\mathcal{R}}}^\psi$, the policy update rule induces a perturbed dynamic. Specifically, for any agent $i$ and deviation $\phi \in \Phi_i$, the uniform approximation bound implies over finite time $T$:
\begin{equation}
    \frac{1}{T}\sum_{t=1}^{T} \mathcal{R}^{\mathrm{true}}_{i,t}(\phi)
    \le
    \frac{1}{T}\sum_{t=1}^{T} \hat{\mathcal{R}}^{\psi}_{i,t}(\phi) + \epsilon_c.
\end{equation}
Taking the maximum over $\phi$ and considering only the positive deviation incentives, the average true regret is bounded by:
\begin{equation}
    \left[ \max_{\phi \in \Phi_i} \frac{1}{T}\sum_{t=1}^{T} \mathcal{R}^{\mathrm{true}}_{i,t}(\phi) \right]_+
    \le
    \left[ \max_{\phi \in \Phi_i} \frac{1}{T}\sum_{t=1}^{T} \hat{\mathcal{R}}^{\psi}_{i,t}(\phi) \right]_+ + \epsilon_c.
\end{equation}
Recall that our Lagrangian objective explicitly constrains the $L_2$ norm of the predicted positive regret vector, $\mathcal{M}_i \le \delta_{\text{regret}}$. Because bounding the $L_2$ norm strictly bounds the maximum element ($\|x\|_\infty \le \|x\|_2$), the learning dynamics successfully minimize the first term on the right-hand side up to the slack $\delta_{\text{regret}}$. 
Furthermore, because our actor utilizes a Softmax approximation with temperature $\tau$ (as defined in the main text) rather than a strict step-function Regret Matching operator ($\pi_{t+1} \propto [\hat{\boldsymbol{\mathcal{R}}}^\psi_t]^+$), an additional entropy-induced approximation error proportional to $\mathcal{O}(\tau \log |\mathcal{A}_i|)$ is introduced. 
Taking the asymptotic limit $T \to \infty$, the $\mathcal{O}(1/\sqrt{T})$ convergence term vanishes. Thus, the empirical distribution converges to the set of $\epsilon$-CE with $\epsilon = \mathcal{O}(\epsilon_c + \delta_{\text{regret}} + \tau \log |\mathcal{A}_i|)$. When $\tau$ is annealed to zero, this last term vanishes.
\end{proof}

\section{Detailed Architecture and Implementation}
\label{app:implementation}

\subsection{Centralized Attention Critic}
To avoid explicit enumeration of counterfactual rollouts, we implement a centralized attention critic that predicts vector-valued regret estimates in a single forward pass.

\paragraph{Regret-Head Justification \& Gradient Flow.}
While the raw counterfactual advantage can be computed as $Q_i(s,(a_i',\mathbf a_{-i})) - Q_i(s,\mathbf a)$, the regret constraint depends on its positive part: $\left[Q_i(s,(a_i',\mathbf a_{-i})) - Q_i(s,\mathbf a)\right]_+$. 
Accordingly, the explicit regret head predicts a smoothed non-negative approximation of this positive-part regret.
Crucially, as mentioned in the main text, we ensure end-to-end differentiability:
\begin{enumerate}
    \item \textbf{Softplus Activation:} The Regret-Head uses a \texttt{Softplus} activation to strictly enforce non-negativity ($\hat{\mathcal{R}} \ge 0$) while maintaining smooth gradients, unlike ReLU which has zero gradients in the negative regime.
    \item \textbf{Gumbel-Softmax:} The actor's policy outputs are sampled via the Gumbel-Softmax trick (Temperature $\tau=1.0$, annealed to $0.1$). This allows the gradient from the stability constraint $\mathcal{M}_i = \|\hat{\boldsymbol{\mathcal{R}}}_i\|_2$ to backpropagate through the critic to the actor.
\end{enumerate}

\paragraph{Network Initialization.}
To prevent value explosion in the early stages of training, we apply a specific initialization to the Q-Head's output layer: weights are scaled by $0.01$ and biases are initialized to $0$. This ensures that the initial Q-values are near zero. For the regret head, we initialize the final bias to a small negative value so that the Softplus output starts close to zero, allowing the agents to begin with unbiased exploration.

\paragraph{FiLM Modulation.}
To handle non-stationarity, the \texttt{Regret Coordinator} module uses Feature-wise Linear Modulation (FiLM). The cumulative regret $\boldsymbol{\mathcal{R}}^{\text{cum}}$ is processed by an MLP to generate scale ($\gamma_{\text{FiLM}}$) and shift ($\beta_{\text{FiLM}}$) parameters. Let $F$ denote the intermediate feature representation of the global state and joint action. The modulated feature $F'$ is computed as:
\begin{equation}
    F' = \gamma_{\text{FiLM}}(\boldsymbol{\mathcal{R}}^{\text{cum}}) \odot F + \beta_{\text{FiLM}}(\boldsymbol{\mathcal{R}}^{\text{cum}})
\end{equation}
This allows the critic to dynamically adapt its attention focus based on the current distance to equilibrium.

\subsection{Actor Architecture}
\begin{itemize}
    \item \textbf{Vector Inputs (MPE):} MLP with 2 hidden layers (128 units), ReLU activation.
    \item \textbf{Pixel Inputs (Melting Pot):} Standard Nature-CNN architecture.
    \begin{itemize}
        \item Input: $68 \times 68$ Grayscale (Stacked 4 frames)
        \item Conv1: 32 filters, $8 \times 8$, stride 4
        \item Conv2: 64 filters, $4 \times 4$, stride 2
        \item Conv3: 64 filters, $3 \times 3$, stride 1
        \item FC: 1024 $\to$ 128 $\to$ Action Dim
    \end{itemize}
\end{itemize}

\subsection{Lagrangian Auto-Tuning and Exploration}
We automate the selection of penalty coefficients $\alpha$ via Dual Gradient Descent:
\begin{itemize}
    \item \textbf{Fairness:} $\alpha^{\text{fair}}$ increases only when Regret Norm $> \delta_{\text{regret}}$ (0.15).
    \item \textbf{Entropy Annealing:} In our implementation, the entropy target decays from a ratio of 0.8 to 0.05 over the first 10\% of episodes to facilitate the transition from Discovery to Convergence.
    \item \textbf{Regret Exploration Bonus:} As listed in the hyperparameters, we utilize an exploration bonus ($\lambda_{\text{regret}}$) to scale intrinsic exploration signals during the initial discovery phase.
\end{itemize}

\section{Environment Settings}
\label{app:env}

\subsection{Matrix Games Payoff}
\label{app:matrix_payoff}

For reproducibility, we explicitly detail the payoff matrices used in Section~5.1. The values represent the joint reward $(r_1, r_2)$ for the row player (Agent 1) and column player (Agent 2).

\begin{table}[h]
\centering
\caption{Payoff Matrices for Iterated Matrix Games. Actions are denoted as \textbf{C} (Cooperate) and \textbf{D} (Defect).}
\label{tab:matrix_payoff}
\renewcommand{\arraystretch}{1.2}
\begin{tabular}{lcccc}
\toprule
\multirow{2}{*}{\textbf{Game}} & \multirow{2}{*}{\textbf{Ag 1 \textbackslash Ag 2}} & \multicolumn{2}{c}{\textbf{Agent 2 Action}} \\
\cmidrule(lr){3-4}
 & & \textbf{C} & \textbf{D} \\
\midrule
\multirow{2}{*}{Prisoner's Dilemma} & \textbf{C} & (3, 3) & (0, 5) \\
 & \textbf{D} & (5, 0) & (1, 1) \\
\midrule
\multirow{2}{*}{Chicken} & \textbf{C} & (3, 3) & (1, 4) \\
 & \textbf{D} & (4, 1) & (0, 0) \\
\midrule
\multirow{2}{*}{Stag Hunt} & \textbf{C (Stag)} & (4, 4) & (0, 3) \\
 & \textbf{D (Hare)} & (3, 0) & (3, 3) \\
\bottomrule
\end{tabular}
\end{table}

\subsection{MPE Zero-Sum Modification}
In \texttt{simple\_adversary}, to ensure a rigorous zero-sum evaluation, we modified the reward structure such that the sum of rewards is strictly zero ($+10$ for the adversary, $-10$ for the agents). Since the modified reward is strictly zero-sum, the total collective reward across both sides is zero by construction and cannot be used as a convergence metric. Instead, we evaluate convergence using side-specific returns, regret, and the reduction of exploitability-like performance gaps. In this setting, $\Phi$-AC successfully stabilizes these metrics without oscillating into a regime where one side continually exploits the other, unlike the baselines.

\subsection{Melting Pot (Harvest)}
We use the \texttt{commons\_harvest\_open} substrate with the following preprocessing:
\begin{itemize}
    \item \textbf{Resolution:} Resized to $68 \times 68$.
    \item \textbf{Frame Stack:} Last 4 frames are stacked to capture temporal dynamics.
    \item \textbf{Action Repeat:} Actions are repeated for 4 steps.
    \item \textbf{Reward:} $+1$ for collecting an apple. Zapping removes a target from the game for 25 steps.
\end{itemize}

\section{Hyperparameters}
\label{app:hyperparams}

We provide detailed hyperparameters for both experimental domains. Note that the simpler MPE environment utilizes larger batch sizes and learning rates compared to the pixel-based Melting Pot environment.

\subsection{MPE and Matrix Games}
Table \ref{tab:hyperparams_mpe} lists the settings used for MPE environments (\texttt{simple\_spread}, \texttt{simple\_tag}, \texttt{simple\_adversary}) and matrix games.

\begin{table}[H]
\centering
\caption{Hyperparameters for MPE \& Matrix Games}
\label{tab:hyperparams_mpe}
\begin{tabular}{ll}
\toprule
\textbf{Parameter} & \textbf{Value} \\
\midrule
\multicolumn{2}{l}{\textit{Optimization}} \\
Total Episodes & 30,000 \\
Episode Length & 25 \\
Actor Learning Rate & $3 \times 10^{-4}$ \\
Critic Learning Rate & $3 \times 10^{-4}$ \\
Batch Size & 1024 \\
Buffer Capacity & 100,000 \\
Discount Factor ($\gamma$) & 0.99 \\
Soft Update ($\tau$) & 0.01 \\
\midrule
\multicolumn{2}{l}{\textit{$\Phi$-AC Specific}} \\
Regret Exploration Bonus ($\lambda_{\text{regret}}$) & 1.0 \\
Regret Bias Scale ($\beta$) & 1.0 \\
Fairness Threshold ($\delta_{\text{regret}}$) & 0.15 \\
FiLM Scale & 1.0 \\
Entropy Decay (Ratio) & $1.0 \to 0.05$ \\
Regret EMA Decay & 0.995 \\
\bottomrule
\end{tabular}
\end{table}

\subsection{Melting Pot (Harvest)}
Table \ref{tab:hyperparams_harvest} lists the settings for the sequential social dilemma experiments.

\begin{table}[H]
\centering
\caption{Hyperparameters for Melting Pot (Harvest)}
\label{tab:hyperparams_harvest}
\begin{tabular}{ll}
\toprule
\textbf{Parameter} & \textbf{Value} \\
\midrule
\multicolumn{2}{l}{\textit{Optimization}} \\
Actor Learning Rate & $1 \times 10^{-4}$ \\
Critic Learning Rate & $1 \times 10^{-4}$ \\
Batch Size & 512 \\
Buffer Capacity & 30,000 \\
Discount Factor ($\gamma$) & 0.99 \\
Soft Update ($\tau$) & 0.01 \\
\midrule
\multicolumn{2}{l}{\textit{$\Phi$-AC Specific}} \\
Regret Exploration Bonus ($\lambda_{\text{regret}}$) & 1.0 \\
Regret Bias Scale ($\beta$) & 0.75 \\
Fairness Threshold ($\delta_{\text{regret}}$) & 0.15 \\
FiLM Scale & 1.5 \\
Entropy Decay (Ratio) & $0.8 \to 0.05$ \\
Regret EMA Decay & 0.90 \\
\bottomrule
\end{tabular}
\end{table}
\end{document}